\documentclass[10pt]{article}
\usepackage{color,graphicx}
\usepackage{amsfonts,amsmath,amssymb,graphicx}
\usepackage[english]{babel}
\usepackage{tikz}

{}
\newtheorem{definition}{Definition}{}
\newtheorem{lemma}{Lemma}{}
\newtheorem{theorem}{Theorem}{}

\newenvironment{proof}{\textbf{Proof}}{}

\newtheorem{example}{{\bf Example}} {}

\newcommand\Alg{\textsc{Alg}}
\newcommand\CNN{\textsc{cnn }}
\newcommand\RR{\mathbb{R}}

\newcommand\RRplus{\mathbb{R}^+}
\newcommand\nnabla{\nabla}
\newcommand\snabla{\xi}
\newcommand\Boks{\textsc{Box}}

\newcommand\Spheres{\textsc{Spheres}}

\newcommand\Argmin{\mathop{\mbox{\rm Argmin}}}
\newcommand\Opt{\mbox{\sc Opt}}
\newcommand\WFA{\textsc{\mbox{\rm WFA}}}
\newcommand\MM{\mathbb{M}}
\newcommand\XX{\mathbb{X}}
\newcommand\YY{\mathbb{Y}}
\newcommand\Gg{\mathcal{G}}
\newcommand\Ff{\mathcal{F}}
\newcommand\Hh{\mathcal{H}}
\newcommand\Sl{Sl}
\newcommand\yy{y^*}
\newcommand\dd{\partial}
\newcommand\roo{\rho}

\newcommand\BST{\textsc{bst}}

\begin{document}

\title{The generalized work function algorithm is competitive for the generalized 2-server problem}

\author{Ren\'e Sitters\thanks{VU University Amsterdam, the Netherlands ({\tt r.a.sitters@vu.nl}), and CWI Amsterdam, the Netherlands ({\tt sitters@cwi.nl}).}}

\maketitle

\begin{abstract}
The generalized 2-server problem is an online optimization problem where a sequence of requests has to be served at minimal cost. Requests arrive
one by one and need to be served instantly by at least one of two servers. We consider the general model where the cost function of the two
servers may be different. Formally, each server moves in its own metric space and a request consists of one point in each metric space. It is
served by moving one of the two servers to its request point. Requests have to be served without knowledge of the future requests. The objective
is to minimize the total traveled distance. The special case where both servers move on the real line is known as the \CNN problem.

We show that the generalized work function algorithm, $\WFA_{\lambda}$, is constant competitive for the generalized 2-server problem. Further, we
give an outline for a possible extension to $k\geqslant 2$ servers and discuss the applicability of our techniques and of the work function algorithm in
general. We conclude with a discussion on several open problems in online optimization.

\end{abstract}



\pagestyle{myheadings} \thispagestyle{plain} \markboth{REN\'E SITTERS}{THE GENERALIZED WORK FUNCTION ALGORITHM}

\section{Introduction}
The work function algorithm is a generic algorithm for online optimization problems. For many problems, it gives the optimal competitive ratio or
it is conjectured to be optimal. For example, it has the best known ratio of $2k-1$  for the $k$-server problem~\cite{KouPap95A}, which is
probably the most appealing and well-studied problem in online optimization, and the work function algorithm is conjectured to have an optimal
ratio of $k$. There are many papers that deal with this classical work function algorithm. More powerful, but less known is the \emph{generalized}
work function algorithm, $\WFA_{\lambda}$, which is the standard work function algorithm with an additional parameter $\lambda$. A result by
Burley~\cite{Bu96} shows that the generalized algorithm can indeed be strictly more powerful than the standard work function algorithm.

The (generalized) work function algorithm may be computationally expensive and pretty hard to analyze, but things can be much better for special
cases. For example, the simple doubling algorithm for the cow path problem is mimed by the generalized work function algorithm $\WFA_{0.5}$.
Another example is the (optimal) move-to-front algorithm for the list update problem which can be seen as the work function algorithm $\WFA_{1}$.
The running time of the work function algorithm very much depends on the complexity of computing  offline solutions. For example, the work
function algorithm for traversing layered graphs can be implemented in linear time while its analysis is quite involved~\cite{Bu96}. Further, the
performance of the work function algorithm may be much better in practice than what is guaranteed in theory (See for
example~\cite{BlumBurch2000}). For some problems, the work function algorithm is optimal but there are more efficient alternatives. For example,
it is $k$-competitive for weighted caching~\cite{BartalKoutsoupias2004} but the elegant Double Coverage algorithm~\cite{ChrLar91A} has the same
optimal ratio. However, the DC-algorithm is not extendable to arbitrary metric spaces. For some hard problems, the work function algorithm is
basically the only algorithm known. Examples are the (deterministic) $k$-server problem and the generalized 2-server problem that we discuss in
this paper.

We say that an algorithm $\Alg$ for an online minimization problem is $c$-\emph{competitive} ($c\geqslant 1$) if there is a constant $c_0$ such
that for every instance $\mathcal{I}$ of the problem, the algorithms cost $\Alg({\mathcal{I}})$ and the optimal cost $\Opt({\mathcal{I}})$ satisfy
\[\Alg({\mathcal{I}})\leqslant c\cdot \Opt({\mathcal{I}})+c_0.\]
The competitive ratio of the algorithm is the infimum over $c$ such that $\Alg$ is $c$-competitive. The competitive ratio of the online
minimization problem is the infimum over $c$ such that there is a $c$-competitive algorithm.

In the \emph{generalized 2-server problem} we are given a server, whom we will call the $\XX$-server, moving in a symmetric metric space
$\mathbb{X}$, and a server, the $\YY$-server, moving in a symmetric metric space $\mathbb{Y}$. A starting point
$(\mathcal{O}^\mathbb{X},\mathcal{O}^\mathbb{Y})\in \mathbb{X}\times \mathbb{Y}$ is given and requests $(x,y)\in \mathbb{X}\times \mathbb{Y}$ are
presented on-line one by one. Requests are served by moving one of the servers to the corresponding point in its metric space and the choice of
which server to move is made without knowledge of the future requests. The objective is to minimize the sum of the distances traveled by the two
servers. The special case $\XX=\YY=\RR$ is known as the \CNN problem~\cite{KouTay00,KoutsoupiasTaylor2004}. This problem can be seen as a single
server moving in $\RR^2$ with the $L_1$-norm and each request is a point in $\RR^2$ which is served if the $x$- or $y$-coordinate of the server
and request coincide.

Research on the \CNN problem started more than ten years ago but despite the simplicity of the problem and its importance for the theory of online
optimization, the problem is still not well-understood. The \CNN problem first appeared in a paper of Koutsoupias and
Taylor~\cite{KouTay00,KoutsoupiasTaylor2004}. They  conjectured that the generalized work function algorithm $\WFA_{\lambda}$ has a constant
competitive ratio\footnote{The $\lambda$ in~\cite{KoutsoupiasTaylor2004} corresponds with $1/\lambda$ in our notation.} for any $\lambda\in
(0,1)$. They also conjectured that the generalized work function algorithm is competitive for the generalized 2-server problem. In this paper we
settle both conjectures. The constant that follows from our proof is large and we do not present an upper bound on its value. Hence, the gap
between known lower and upper bound remains large. The first competitive algorithm was given in~\cite{SittersSP:2003} and in its journal
version~\cite{SittersStougie2006}. The importance of the new result here is that we analyze the generalized work function algorithm which is
applicable to \emph{any} metrical service systems. Our techniques here are more involved than those in~\cite{SittersStougie2006} and are
interesting for online optimization in general. This is discussed in Section~\ref{sec:future}.

As the name suggests, the generalized 2-server problem originates from the \emph{classical} 2-server problem in which $\XX=\YY$ and $x=y$ for
every request, i.e., each request is a point in the metric space and we have to decide which server to move to the requested point. The
$k$-server problem (with $k\geqslant 2$ servers) is one of the most studied problems in online optimization. A recent survey of the $k$-server
problem is given by Koutsoupias~\cite{Koutsoupias2009}. The $k$-server problem on a uniform metric space is the \emph{paging problem}. In the
\emph{weighted} $k$-server problem a weight is assigned to each server (of the classical problem) and the total cost is the weighted sum of the
distances. The weighted $k$-server problem is a special case of the generalized $k$-server problem.

All online optimization problems mentioned in this article belong to the class of \emph{metrical task systems} (For a definition see
Section~\ref{sec:MSS} and~\cite{BoLiSa92}). Given multiple metrical task systems, the \emph{sum problem}~\cite{KoutsoupiasTaylor2004} is again a
metrical task system and is defined as follows: At each step we receive one request for each task system and we have to serve at least one of
those requests. The \CNN problem is the sum of two trivial problems: in both problems there is one server moving on the real line and each request
consists of a single point. Koutsoupias and Taylor~\cite{KoutsoupiasTaylor2004} emphasize the importance of the \CNN problem: `\emph{It is a very
simple sum problem, which may act as a stepping stone towards building a robust (and less ad hoc) theory of online computation}'. Indeed, our
techniques are useful for sum problems in general and we hope it leads to a better insight and hence further simplifications and generalizations
in the theory of online computation.

\subsection{Known competitive ratios}
No \emph{memoryless} (randomized) algorithm can have a finite competitive ratio  for the \CNN problem
\cite{ChrobakSgall2004,KoutsoupiasTaylor2004,Verhoeven2006} while a finite ratio is possible if we are allowed to store the entire given sequence
(or at least the current work function)~\cite{SittersStougie2006,SittersSP:2003}. The algorithms in the latter two papers are complex and the
ratio very high but they do apply to the generalized 2-server problem as well. See~\cite{Chrobak2003} for a review of~\cite{SittersSP:2003}. For
the classical $k$-server problem, the work function algorithm is $(2k-1)$-competitive for any metric space~\cite{KouPap95A} and it is conjectured to
be even $k$-competitive~\cite{ManasseMS1988,KouPap95A}. This famous $k$-\emph{server conjecture} was posed more than two decades ago and is still
open. The competitive ratio of the \emph{weighted} $k$-server problem is much higher.  Fiat and Ricklin~\cite{FiaRic94} prove that for any metric
space with at least $k+1$ points there exists a set of weights such that the competitive ratio of any deterministic algorithm is at least
$k^{\Omega(k)}$. Koutsoupias and Taylor [23] prove that any deterministic online algorithm for the weighted 2-server problem has competitive ratio
at least $6+\sqrt{17}>10.12$ even if the underlying metric space is the line and~\cite{ChrobakSgall2004} shows that any memoryless randomized
algorithm has unbounded competitive ratio in this case. These two lower bounds apply to the \CNN problem as well since it contains the weighted
two server problem on the line as a special case.

\subsection{More special cases and variants}
The \emph{orthogonal} \CNN problem~\cite{IwamaYonezawa2004} is the special case of the \CNN problem in which each request either shares the
$x$-coordinate or the $y$-coordinate with the previous request. Iwama and Yonezawa~\cite{IwamaYonezawa2004} give a $9$-competitive algorithm
and a lower bound of 3 is given in~\cite{AugustineGravin2010}.

In the \emph{continuous} \CNN problem~\cite{AugustineGravin2010}, there is one request which follows a continuous path in $\RR^2$ and the online
server must serve it continuously by aligning either horizontally or vertically. It generalizes the orthogonal version in the sense that any
$c$-competitive algorithm for the continuous problem implies a $c$-competitive algorithm for the orthogonal problem. Augustine and
Gravin~\cite{AugustineGravin2010} give a $6.46$-competitive memoryless algorithm (improving the 9 from~\cite{IwamaYonezawa2004} mentioned
above).

The \emph{axis-bound} \CNN problem was introduced by Iwama and Yonezawa~\cite{IwamaYonezawa2001,IwamaYonezawa2002} and is the special case in
which the server can only move on the $x$- and $y$-axis. They give an upper bound of 9 and a lower bound of $4+\sqrt{5}$. The lower bound was
raised to 9 in~\cite{AusielloBL2004b}. That paper also gives an alternative $9$-competitive algorithm by formulating it as a \emph{two point
request problem}~\cite{Bu96}. Finally, the \emph{box bound} \CNN problem~\cite{AusiellBL2004a} is the restriction in which the server can move
only on the boundary of a rectangle and requests are inside the rectangle. The problem can be transformed into the 4-point request
problem~\cite{AusiellBL2004a}. An upper bound of $88.71$ for the latter problem follows from the paper by Burley~\cite{Bu96}.

For the weighted 2-server problem, the only known competitive algorithm follows from the one for the generalized 2-server problem. For the special
case of a \emph{uniform} metric space (where all distances are 1), Chrobak and Sgall~\cite{ChrobakSgall2004} prove that the  work function
algorithm is 5-competitive and that no better ratio is possible. They also give a 5-competitive randomized, memoryless algorithm for uniform
spaces, and a matching lower bound. Further, they consider  a version of the problem in which a request specifies two points to be covered by the
servers, and the algorithm must decide which server to move to which point. For this version, they show a 9-competitive algorithm and prove that
no better ratio is possible. Finally,  Verhoeven~\cite{Verhoeven2006} shows that no memoryless randomized algorithm can be competitive for the
\CNN problem under an even weaker definition of memoryless than used in~\cite{ChrobakSgall2004} and~\cite{KoutsoupiasTaylor2004}.

\subsection{Metrical task systems and metrical service systems}\label{sec:MSS}

Borodin, Linial, and Saks~\cite{BoLiSa92} introduced the problem of \emph{metrical task systems}, a generalization of all online problems
discussed here. Such system is a pair $\mathcal{S}=(\MM,\mathcal{T})$, where $\MM$ is a metric space and $\mathcal{T}$ a set of tasks. Each task
$\tau\in\mathcal{T}$ is defined by a function $\tau:\MM\rightarrow \RR^+$ which gives for each $s\in \MM$ the cost of serving the task while
being in $s$. In an online instance, the tasks are given one by one and the objective is to minimize the total traveled distance (starting from
given origin $\mathcal{O}$) plus the total service cost. The system is called \emph{unrestricted} if $\mathcal{T}$ consists of all non-negative
real functions on $\MM$. The authors of~\cite{BoLiSa92} show that the competitive ratio is exactly $2m-1$ for the unrestricted metrical task
system on any metric space on $m$ points.

A restricted model is that of \emph{metrical service systems}, introduced in~\cite{ChLa92},\cite{ChrLar92D} and~\cite{MaMcSl90}.
(In~\cite{MaMcSl90} it is called \emph{forcing task systems}.) Such a system is a pair $\mathcal{S}=(\MM,\mathcal{R})$, where $\MM$ is a metric
space and $\mathcal{R}$ a set of requests where each request $r\in\mathcal{R}$ is a subset of $\MM$. The system is called \emph{unrestricted} if
$\mathcal{R}$ consists of all subsets of $\MM$. Metrical service systems correspond to metrical task systems for which $\tau:\MM\rightarrow
\{0,\infty\}$ for each task $\tau$. Manasse et al.~\cite{MaMcSl90} give an optimal $(m-1)$-competitive algorithm for the unrestricted metrical
service system on any metric space on $m$ points.

The generalized 2-server problem is a metrical service system: There is one server moving in the product space $\MM=\mathbb{X}\times \mathbb{Y}$
and any pair $(x,y)\in \mathbb{X}\times \mathbb{Y}$ defines a request $r(x,y)=\{\{x\}\! \times\!  \mathbb{Y}\} \cup \{\mathbb{X}\! \times\!
\{y\}\}\subset \MM$. The distance between points $(x_1,y_1)$ and $(x_2,y_2)$ in $\mathbb{X}\times \mathbb{Y}$ is
$d((x_1,y_1),(x_2,y_2))=d^{\XX}(x_1,x_2)+d^{\YY}(y_1,y_2)$, where $d^{\XX}$ and $d^{\YY}$ are the distance functions of the metric spaces
$\mathbb{X}$ and $\mathbb{Y}$.

The work function algorithm is optimal for metrical task and metrical service systems in the sense that it is, respectively, $(2m-1)$- and $(m-1)$-competitive on any metric space of at most $m$ points~\cite{ChrobakLarmore1998}. This is not of direct use for the \CNN problem since the metric
space, $\RR^2$, has an unbounded number of points.

\subsection{The work function algorithm: $\WFA_{\lambda}$}
The work function algorithm appeared for the first time in~\cite{ChLa92} but was discovered independently by others (see~\cite{KouPap95A}). We
use it here only for metrical service systems but it works the same for metrical task systems.
\begin{definition}
Given a metrical service system $\mathcal{S}=(\MM,\mathcal{R})$ and origin $\mathcal{O}\in \MM$, and given a request sequence $\sigma$, the work
function $W_{\sigma}:\MM\rightarrow \RRplus$ is defined as follows. For any point $s\in \MM$, $W_{\sigma}(s)$ is the length of the shortest path
that starts in $\mathcal{O}$, ends in $s$ and serves $\sigma$.
\end{definition}

We assume here that the work function is well-defined (which is always true if the metric space is finite). Thus, we assume that for any
$\sigma=r_1,\ldots,r_n$ and any point $s\in \MM$ there are points $s_i\in r_i$ $(i=1,\ldots, n)$ such that
$d(\mathcal{O},s_1)+d(s_1,s_2)+\cdots+d(s_{n-1},s_n)+d(s_n,s)\leqslant d(\mathcal{O},t_1)+d(t_1,t_2)+\cdots+d(t_{n-1},t_n)+d(t_n,s)$ for any set
of points  $t_i\in r_i$ $(i=1,\ldots, n)$. Clearly, the work function is well-defined for the generalized 2-server problem since we may assume
that for each $s_i$, both its \emph{coordinates} are from requests given so far. See~\cite{ChrLar92D} for a sufficient condition for the work
function to be well-defined.

For a work function $W_{\sigma}$ we say that point $s$ is \emph{dominated} by point $t$ if $W_{\sigma}(s)=W_{\sigma}(t)+d(s,t)$. We define the
\emph{support} of $W_{\sigma}$ as
\[\textrm{supp}(W_{\sigma})=\{s\in \MM : s\textrm{ is not dominated by any other point}\}.\]
Let $\sigma r$ denote the sequence $\sigma$ followed by request $r$. If $W_{\sigma r}$ is a well-defined work function then supp$(W_{\sigma
r})\subseteq r$ since for any point $s\notin r$ there exists a point $t\in r$ such that $W_{\sigma r}(s)=W_{\sigma r}(t)+ d(s,t)$. For more
properties and a deeper analysis of the work function (algorithm) see for example \cite{BorodinElYaniv1998Book},\cite{Bu96},
and~\cite{Koutsoupias2009}.

The generalized work function algorithm is a work function-based algorithm parameterized by some constant $\lambda \in (0,1]$. We denote it by
$\WFA_{\lambda}$.

\begin{definition}\label{def:WFAlgo}
For any request sequence $\sigma$ and any new request $r$, the \emph{generalized work function algorithm} $\WFA_{\lambda}$ moves the server from
the position $s$ it had after serving $\sigma$ to any point
\begin{equation}\label{def:WFA}
s'\in {\rm Argmin}_{t\in \MM}
\{  W_{\sigma r}(t)+ \lambda d(s,t)\} .
\end{equation}
\end{definition}
This minimum may not be well-defined if the request $r$ contains infinitely many points of the metric space. This is no problem for the
generalized 2-server problem since the minimum is attained for some $t$ with both coordinates of the given requests~\cite{SittersStougie2006}.
From~\eqref{def:WFA}, we see that
\[W_{\sigma r}(s')+ \lambda d(s,s') \leqslant W_{\sigma r}(t)+\lambda d(s,t)\ \text{ for any point $t\in \MM$}.\]
Using the triangle inequality, we get that for any $t\in \MM$
\begin{eqnarray}\label{eq:lWFA2}
  W_{\sigma r}(s')\leqslant W_{\sigma r}(t)+  \lambda (d(s,t)-d(s,s')) \leqslant W_{\sigma r}(t)+  \lambda d(s',t).
\end{eqnarray}
If $\lambda<1$ then~(\ref{eq:lWFA2}) implies that $s'$ is not dominated by any other point, whence $s'\in \textrm{supp}(W_{\sigma r})\subseteq r$.
We see that if the moves of $\WFA_{\lambda}$ are well-defined then the choice of $\lambda<1$ \emph{ensures} that the point $s'$ always serves the
last request and we may \emph{replace} $t\in \MM$ by ${t\in r}$ in Definition~\ref{def:WFAlgo}.

For $\lambda=0$, the generalized work function algorithm corresponds to the algorithm that always moves to the endpoint of an optimal solution,
and for $\lambda=\infty$ it corresponds to the greedy algorithm (if we take $t\in r$ instead of $t\in \MM$ in~\eqref{def:WFA}). The standard
work function algorithm has $\lambda=1$ and was first used in~\cite{ChLa92} and has been studied extensively. The general form was defined
in~\cite{ChLa92} as well but was used only shortly after in~\cite{ChrLar92D} where it is called  the $\lambda$-Cheap-and-Lazy strategy. They
show that $\WFA_{\lambda}$ with $\lambda=1/3$ is optimal for the 2-point request problem. Burley generalized this and showed that
$\WFA_{\lambda}$ is $O(k2^k)$-competitive for the $k$-\emph{point request problem} (where $\lambda$ depends on $k$).

In most papers, $\lambda$ is placed before $W_{\sigma r}$ in~\eqref{def:WFA} instead of before $d(s,t)$, as we do here. Also, sometimes $\alpha$
is used instead of $\lambda$. For example, Burley~\cite{Bu96} uses $\alpha>1$ and the following definition of the work function algorithm: $s'\in
{\rm Argmin}_{t\in \MM} \{  \alpha W_{\sigma r}(t)+ d(s,t)\}$. Replacing $\alpha$ by $1/\lambda$ matches our definition. Our choice was partly
 for an aesthetical reason: Now, the term $\lambda$ appears much more often in the paper than the term $1/\lambda$. But also in the
definitions of the \emph{extended cost} and \emph{slack function} (Section~\ref{sec:prelimanaries}), using $\lambda<1$ seems the natural choice.

\subsection{Paper outline and proof sketch}
The main part of this paper is devoted to the \CNN problem (Theorem~\ref{th:theorem1}). The generalization to arbitrary metric spaces
(Theorem~\ref{th:theorem2}) is more complex and we do this in a separate section. The proof of Theorem ~\ref{th:theorem1} is based on no less
than 21 lemmas.  To obtain a better insight in the relation between lemmas we mention after each lemma where it is used. The proof of
Theorem~\ref{th:theorem2} uses exactly the same lemmas (only some constants are different) and we indicate how to adjust the proofs of these
lemmas.

Before giving a sketch of the proof, we give a brief outline of the paper. In Section~\ref{sec:prelimanaries} we list some properties of the work
function algorithm. These hold for any metrical service system and can be found in several other papers,
e.g.~\cite{Bu96,ChrobakLarmore1998,KouPap95A}. Further, we introduce the closely related \emph{slack function} and list some of its properties. In
Section~\ref{sec: cnn problem} we present our potential function for the \CNN problem together with some of its properties. Although the potential
function is defined for the \CNN problem, the theory in Sections~\ref{sec:potential cnn} and~\ref{sec:prop potential} applies to any metrical
service system on $\RR^2$. In Section~\ref{sec:prop cnn} we state some properties of the \CNN problem which do not depend on the potential and in
Section ~\ref{sec:pot applied} we put everything together and apply the potential to the \CNN problem. In Section~\ref{sec:2-server} we show how
to modify the proof  for general metric spaces, i.e., we prove that the generalized work function algorithm is constant competitive for the
generalized 2-server problem. In Section~\ref{sec:higher dim} we give a sketch of a possible extension to higher dimensions. Finally, in
Section~\ref{sec:future} we discuss several open problems in online optimization.

There are several reasons for giving a separate \CNN proof. First, the reader has the option of just reading the \CNN proof and skip the more
difficult general proof. Nevertheless, we believe that the generalization is relatively easy to digest once the reader has worked through the
\CNN proof and it may be even easier this way than when we would present only the general proof.  One reason is that the \CNN problem can be
seen as moving points in the Euclidean plane which makes the proof easier to visualize than the proof for the general case.

Our potential function has a long description and may seem unintuitive at first. It is a linear combination of two functions: $\Ff$ and $\Gg$.
Function $\Ff$ is a special case of the potential function that was used in~\cite{SittersStougie2006}  to give the first constant competitive
algorithm for the generalized 2-server problem. When we use only $\Ff$ as our potential function and follow the line of proof that we use here,
then the analysis fails. Taking $\Gg$ as potential function does not work either. However, the two functions are in a way complementary and if we
take a linear combination of the two functions then the proof goes through.

Next, we give a very short technical sketch of the proof, which applies to both the \CNN problem and the general problem. This part can be ignored
but it may be very helpful for readers that are familiar with analysis of the work function algorithm. Definitions and formulas given here are
presented in more detail later.

The potential function $\Phi_{\sigma}$ assigns a real value to each request sequence $\sigma$. It has the following form:
\[\Phi_{\sigma}=(1-\gamma)\min_{s_1,s_2,s_3\in\MM}\Ff_{\sigma}(s_1,s_2,s_3)+\gamma\min_{s_1,s_2,s_3\in\MM}\Gg_{\sigma}(s_1,s_2,s_3).\] The
functions $\Ff_{\sigma}:\MM^3\rightarrow \RR$ and $\Gg_{\sigma}: \MM^3\rightarrow \RR$ depend on the sequence $\sigma$. Further,
$\MM=\XX\times\YY$ and $\gamma\in (0,1)$ is a constant. The initial value is zero and in general it is upper bounded by the optimal value of the
sequence so far, i.e., $\Phi_{\sigma} \leqslant \Opt_{\sigma}$. We consider two arbitrary, subsequent requests $r'$ and $r''$ and show that the increase
$\Phi''-\Phi'$ of the potential function for the new request $r''$  is at least some constant $c$ times the so called \emph{extended cost} for
$r''$, denoted by $\nnabla_{r''}$ (Definition~\ref{def:nabla}):
\begin{equation}\label{eq:proof sketch} \Phi''-\Phi'\geqslant c \nnabla_{r''}.\end{equation}
Then, taking the sum over all requests in the entire sequence $\roo$ of the instance~\footnote{In this paper, $\roo$ always refers to the entire
given sequence, i.e., no requests are given after $\roo$. We mainly use $\sigma$ otherwise.}, we find that the total increase in the potential
function is at least $c$ times the total extended cost, denoted by $\nnabla_{\roo}$. We get $\nnabla_{\roo}\leqslant (1/c)\Phi_{\roo}\leqslant (1/c)\Opt_{\roo}$. Proof of competitiveness then follows directly from Lemma~\ref{lem:nabla sigma}.

We now give some more details of Equation~\eqref{eq:proof sketch}. Let $\sigma'$ be a request sequence which ends with $r'$. It is followed by
$r''$ and we denote $\sigma''=\sigma' r''$.  Let $s_1,s_2,s_3$ be a minimizer of $\Ff_{\sigma''}$. By construction of $\Ff$, all three points will
serve the last request $r''$. (The same holds for $\Gg_{\sigma''}$.) We distinguish between \emph{Case A}: $|\{s_1,s_2,s_3\}|\leqslant 2$, and \emph{Case B}: $|\{s_1,s_2,s_3\}|= 3$.  In the following, $c_1,c_2,c_3,c_4>0$ are specific constants depending on $\lambda$. For Case A, we show that \[\min\Ff_{\sigma''}-\min\Ff_{\sigma'}\geqslant c_1\nnabla_{r''}\ \text{ and }\ \min\Gg_{\sigma''}- \min\Gg_{\sigma'}\geqslant 0,\] where
$\nnabla_{r''}$ is the extended cost for $r''$ w.r.t. $\sigma'$. Hence, the increase for $\Phi$ is at least $(1-\gamma)c_1\nnabla_{r''}$ and
~\eqref{eq:proof sketch} holds with $c=(1-\gamma)c_1$. Note that, if we were always in Case A, then there would be no need for function $\Gg$. The more difficult part of the proof is Case B. In that case, it is easy to show that the increase for the minimum of function $\Ff$ is 
\begin{eqnarray*}
\min\Ff_{\sigma''}-\min\Ff_{\sigma'}&\geqslant& c_2\min\{\dd x,\dd y\},
\end{eqnarray*}
where $\dd x=d^{\XX}(x',x'')$ and $\dd y=d^{\YY}(y',y'')$. This is not enough to prove~\eqref{eq:proof sketch} if $\min\{\dd x,\dd y\}\ll \nnabla_{r''}$. So, let us consider the extreme case that $\min\{\dd x,\dd y\}=0$. Intuitively, we should be fine if we can handle this. The function $\Gg$ was designed exactly for this case. More precisely, we show that if
$\min\{\dd x,\dd y\}=0$, then 
\begin{equation}\label{eq:sketch1}
\min\Gg_{\sigma''}- \min\Gg_{\sigma'}\geqslant c_3\nnabla_{r''}.
\end{equation} 
Hence,  the increase for $\Phi$ is at least $\gamma c_3\nnabla_{r''}$ and~\eqref{eq:proof sketch} holds with $c=\gamma c_3$. In general, we prove that
\begin{eqnarray*}
\min\Gg_{\sigma''}- \min\Gg_{\sigma'}&\geqslant& c_3\nnabla_{r''}-c_4\min\{\dd x,\dd y\}.
\end{eqnarray*}
The increase in $\Phi$ becomes  at least
\[\gamma c_3\nnabla_{r''}+((1-\gamma)c_2-\gamma c_4)\min\{\dd x,\dd y\},\]
which becomes at least  $\gamma c_3\nnabla_{r''}$ by choosing $\gamma< c_2/(c_2+c_4)$. Again,~\eqref{eq:proof sketch} applies with $c=\gamma c_3$.

Finally, a few words on how to prove~\eqref{eq:sketch1}. Let $(s_1,s_2,s_3)$ be a minimizer of $\Gg_{\sigma''}$. Remember that, for function $\Ff$, we argued that we are fine if the cardinality of $\{s_1,s_2,s_3\}$ is at most 2. But for function $\Gg$ we can enforce this situation by using the fact that the two subsequent requests are aligned. Say that $\dd y=0$. It will turn out that the only interesting case is when $s_1,s_2,s_3$ are all on the line $y=y'=y''$. (Since other wise,~\eqref{eq:sketch1} will follow almost directly.) For this situation, we prove (see Lemma~\ref{lem:convex containment}) that one of the three points is redundant in the sense that there are points $u_1,u_2\in \{s_1,s_2,s_3\}$ such that $\Gg_{\sigma''}(u_1,u_2,u_2)=\Gg_{\sigma''}(s_1,s_2,s_3)$. Then,~\eqref{eq:sketch1} is proven in a similar was as is done for $\Ff$ in Case A.

Summarizing, function $\Ff$ works fine as a potential function on its own, except for the case that $\min\{\dd x,\dd y\}\approx 0$. In a way, the difficult part is reduced to an easier situation where the two subsequent requests are on a line and we designed a potential function $\Gg$ that takes care of this situation. In other words, the difficult case in the proof is reduced to a problem of lower dimension. This insight led to the generalization to higher dimensions discussed in Section~\ref{sec:higher dim}.
\bigskip

\section{Preliminaries}\label{sec:prelimanaries}
This section applies to any metrical service system. For the analysis of the generalized work function algorithm we make extensive use of two
concepts: \emph{extended cost} and \emph{slack}. The first is an amortized cost of the (general) work function algorithm. It was introduced
together with the work function algorithm in \cite{ChLa92} (where it is called pseudo-cost) and has been used in every analysis of the work
function algorithm. The slack function was defined by Burley~\cite{Bu96} and was also used in~\cite{SittersStougie2006}. Its definition comes
naturally with that of extended cost and its use enhances the analysis.

\subsection{The extended cost}
\begin{definition}\label{def:nabla}
For request sequence $\sigma$ and request $r$, the extended cost for $r$ is
\[\nnabla_r(W_{\sigma})=\max_{s\in \MM}\ \min_{t\in r}\ \left[W_{\sigma}(t)+\lambda d(s,t)-W_{\sigma}(s)\right].\] For $\roo=r_1 r_2\cdots r_n$, we
define the total extended cost as $\nnabla_{\roo}=\sum_{i=1}^{n}\nnabla_{r_i}(W_{r_1\cdots r_{i-1}})$.
\end{definition}

The definition of extended cost matches that in~\cite{SittersStougie2006} and matches the commonly used extended cost in case $\lambda=1$. It also
matches the definition by Burley~\cite{Bu96}, although the notation is quite different. The intuition behind extended cost becomes clear from the
following lemma and its proof.

\smallskip\begin{lemma}\label{lem:nabla sigma} Let $\nnabla_{\roo}$ be the total extended cost of sequence $\roo$. If
$\nnabla_{\roo}\leqslant c\Opt_{\roo}$ for some constant $c$ and any request sequence $\roo$, then $\WFA_{\lambda}$ is
$(c-1)/\lambda$-competitive. (Used in proof of Theorem~\ref{th:theorem1})\end{lemma}
\begin{proof}
Assume the online server is in point $s'$ after it served the initial sequence $\sigma$ and moves to $t'$ to serve a new request $r$. Since we
maximize over $s\in \MM$ in Definition~\ref{def:nabla} we have
\begin{eqnarray*}
\nnabla_r(W_{\sigma})&\geqslant &\min_{t\in r}\left[W_{\sigma}(t)+\lambda d(s',t)-W_{\sigma}(s')\right]\\
 &=&\min_{t\in r}\left[W_{\sigma r}(t)+\lambda d(s',t)-W_{\sigma}(s')\right]
\end{eqnarray*}
By definition of $\WFA_{\lambda}$, the minimum in the right side is attained for $t=t'$. Therefore,
\[\nnabla_r(W_{\sigma})\geqslant W_{\sigma r}(t')+\lambda d(s',t')-W_{\sigma}(s').\]
Rewriting we get
\[d(s',t')\leqslant \frac{1}{\lambda}\left(\nnabla_r(W_{\sigma})+W_{\sigma}(s') -W_{\sigma r}(t')\right).\]
This gives an upper bound for the cost of $\WFA_{\lambda}$ for serving some single request $r$. Let $q$ be the point where the algorithm ends
after serving $\roo$. Summing up over all requests in $\roo$ we get that the total cost for $\WFA_{\lambda}$ is at most
\[\frac{1}{\lambda}(\nnabla_{\roo}+W_{\epsilon}(\mathcal{O})- W_{\roo}(q)) \leqslant \frac{1}{\lambda}(\nnabla_{\roo}-\Opt_{\roo})\leqslant\frac{1}{\lambda}(c-1)\Opt_{\roo}.\]  \end{proof}

\subsection{The slack function}
We use the concept of the {\em slack} of a point relative to another point. Intuitively, the slack of a point $s$ with respect to a point $t$ is
the amount that the work function value in $s$ can increase before the generalized work function algorithm moves from $s$ to $t$. More precisely,
the generalized work function algorithm, being in point $s$ after serving sequence $\sigma$, moves away from $s$ after a new request $r$ is given
if there is a point $t$ such that $W_{\sigma r}(t)+\lambda d(s,t)\leqslant W_{\sigma r}(s)$. The slack is the difference between the left and
right side of this inequality. More generally, we define the slack of a point with respect to a subset of $\MM$. See Figure~\ref{fig:slack}.
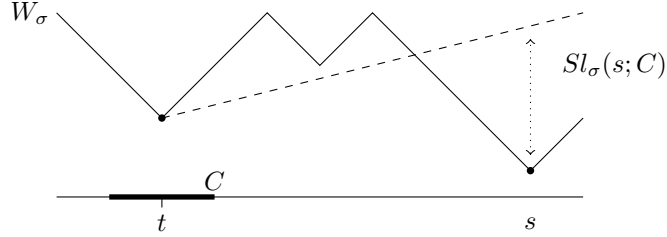
\begin{figure}
\center
\begin{tikzpicture}[scale=0.7]
\draw [line width=2pt] (1,0.5)--(3,0.5);
\draw (0,0.5)--(10,0.5);
\draw (2,0.3)--(2,0.5);
\draw (0,4)--(2,2)--(4,4)--(5,3)--(6,4)--(9,1)--(10,2);
\draw [dashed](2,2)--(10,4);
\draw [<->] [dotted](9,1.3)--(9,3.5);
\fill (canvas cs:x=2cm,y=2cm) circle (2pt);
\fill (canvas cs:x=9cm,y=1cm) circle (2pt);
\node at (2,0) {$t$};
\node at (9,0) {$s$};
\node at (3,0.8) {$C$};
\node at (-0.5,4) {$W_{\sigma}$};
\node at (10.6,3) {$\Sl_{\sigma}(s;C)$};
\end{tikzpicture}
\caption{The slack of $s$ with respect to line segment $C$ is attained in $t\in C$.}
\label{fig:slack}\end{figure}
\begin{definition}
Given a request sequence $\sigma$, we define the slack of a point $s\in \MM$ with respect to a (possibly infinite) set of points $C\subseteq\MM$ as
\[\Sl_{\sigma}(s;C)=\min\limits_{t \in C}\{W_{\sigma}(t)+\lambda d(s,t)\}-W_{\sigma}(s).\]
If $C$ contains only one point $t$ then we simply write $\Sl_{\sigma}(s;t)$ instead of $\Sl_{\sigma}(s;\{t\})$.
\end{definition}
If $C$ is a closed subset of $\MM$ then the minimum is well-defined.

Using the slack function makes the proof shorter and more intuitive. For example, we can rewrite the extended cost, $\nnabla_{r}(W_{\sigma})$, for
request sequence $\sigma$ and new request $r$ in terms of the slack function.
\begin{equation}\label{eq:nabla for r rewritten}
\nnabla_{r}(W_{\sigma})= \max_{s\in \MM}\{\min_{t\in r} \{W_{\sigma}(t)+\lambda d(s,t)\} -W_{\sigma}(s)\}=\max_{s\in \MM} \Sl_{\sigma}(s;r).
\end{equation}
In the remainder of this section, we list some properties of the slack function. The first property~\eqref{eq:Slack 1+lambda} follows directly from
its definition and from the work function being Lipschitz continuous with constant 1. For any $s,t\in \MM$
\begin{equation}\label{eq:Slack 1+lambda}
\Sl_{\sigma}(s;t)\leqslant (1+\lambda)d(s,t).
\end{equation}
The next lemma also follows directly from the definition of slack.
\smallskip\begin{lemma}\label{lem:subsetslack}
If $C_1\subseteq C_2\subseteq \MM$, then for any $s\in\MM$ we have  $\Sl_{\sigma}(s;C_1)\geqslant \Sl_{\sigma}(s;C_2)$. (Used in proof of many
lemmas.)\end{lemma}

The lemma above is mostly used in the form: $t\in C\subseteq \MM$ implies $\Sl_{\sigma}(s;t)\geqslant \Sl_{\sigma}(s;C)$.

\smallskip\begin{lemma}\label{lem:minimum slack}
For any set of points $C\subset\MM$ there is a point $s\in C$ such that $\Sl_{\sigma}(s;C)= 0$. (Used in proof of Lemma~\ref{lem:F<=G<=H}.)
\end{lemma}
\begin{proof}
Let $s\in\Argmin\{W_{\sigma}(t)\mid t\in C\}$. Then, for any $t\in C$:
\[
\Sl_{\sigma}(s;t) =W_{\sigma}(t)+\lambda d(s,t)-W_{\sigma}(s)\geqslant 0.
\]
Clearly, $\Sl_{\sigma}(s;s)=0$. Hence, $\Sl_{\sigma}(s;C)=\min_{t\in C}\Sl_{\sigma}(s;t)=0$.
\end{proof}\bigskip

The next lemma shows a transitivity property of slack.
\smallskip\begin{lemma}\label{lem:transitive slack}
Let  $s_1,s_2,s_3\in \MM$ such $d(s_1,s_2)+d(s_2,s_3)=d(s_1,s_3)$. Then $\Sl_{\sigma}(s_3;s_1)=\Sl_{\sigma}(s_3;s_2)+\Sl_{\sigma}(s_2;s_1)$.
(Used in proof of Lemma~\ref{lem:convex containment}.)
\end{lemma}
\begin{proof}
\[
\begin{array}{rcl}
\Sl_{\sigma}(s_3;s_1)&=&W_{\sigma}(s_1)+\lambda d(s_1,s_3)-W_{\sigma}(s_3)\\
&=&W_{\sigma}(s_1)+\lambda (d(s_1,s_2)+d(s_2,s_3))-W_{\sigma}(s_3)\\
&=&W_{\sigma}(s_1)+\lambda d(s_1,s_2)-W_{\sigma}(s_2)+W_{\sigma}(s_2)+\lambda d(s_2,s_3)-W_{\sigma}(s_3)\\
&=&\Sl_{\sigma}(s_2;s_1)+\Sl_{\sigma}(s_3;s_2).
\end{array}
\]
\end{proof}

The next lemma generalizes Lemma~\ref{lem:subsetslack}.
\smallskip\begin{lemma}\label{lem:slack change C1 to C2}
Let $C_1,C_2\subseteq\MM$ and $\delta\in \RRplus$.If for every point $u_1\in C_1$ there is a point $u_2\in C_2$ with
$d(u_1,u_2)\leqslant\delta$, then for every $s\in\MM$
\[\Sl_{\sigma}(s;C_1)\geqslant \Sl_{\sigma}(s;C_2)-(1+\lambda)\delta.\]
(Used in proof of Lemma~\ref{lem:s dominates t}.)
\end{lemma}
\begin{proof}
Let $u_1\in C_1$ be such that $\Sl_{\sigma}(s;C_1)=\Sl_{\sigma}(s;u_1)$. There is a point $u_2\in C_2$ such that $d(u_1,u_2)\leqslant \delta$.
\[
\begin{array}{rcl}
&&\Sl_{\sigma}(s;C_1)-\Sl_{\sigma}(s;C_2)\\
&=& \Sl_{\sigma}(s;u_1)-\Sl_{\sigma}(s;C_2)\\
&\geqslant & \Sl_{\sigma}(s;u_1)-\Sl_{\sigma}(s;u_2)\\
&=&W_{\sigma}(u_1)+\lambda d(u_1,s)-W_{\sigma}(s)-(W_{\sigma}(u_2)+\lambda d(u_2,s)-W_{\sigma}(s))\\
&=&W_{\sigma}(u_1)-W_{\sigma}(u_2)+\lambda (d(u_1,s)-d(u_2,s))\\
&\geqslant&W_{\sigma}(u_1)-W_{\sigma}(u_2)-\lambda d(u_1,u_2)\\
&\geqslant&- d(u_1,u_2)-\lambda d(u_1,u_2)\\
&=&-(1+\lambda)\delta.
\end{array}
\]
\end{proof}\bigskip

\smallskip\begin{lemma}\label{lem:slack change s to t}
Let  $s,t\in\MM$ and $C\subset\MM$. Then, \begin{enumerate} \item[(a)] $\Sl_{\sigma}(t;C)\geqslant \Sl_{\sigma}(s;C)-(1+\lambda)d(s,t)$, and
\item[(b)]$\Sl_{\sigma}(t;C)\geqslant \Sl_{\sigma}(s;C)+(1-\lambda)d(s,t)$, if $t$ dominates $s$ w.r.t. $\sigma$.
\end{enumerate}
(Follows from Lemma~\ref{lem:subsetslack}. Used in proof of Lemma~\ref{lem:s dominates t},~\ref{lem:nabla dx} and \ref{lem:G increase}.)
\end{lemma}
\begin{proof}
Let $u\in C$ be such that $\Sl_{\sigma}(t;C)=\Sl_{\sigma}(t;u)$. Then,
\[
\begin{array}{rcl}
\Sl_{\sigma}(t;C)&=&\Sl_{\sigma}(t;u)\\
&=&\Sl_{\sigma}(s;u)-\lambda d(u,s)+\lambda d(u,t)+W_{\sigma}(s)-W_{\sigma}(t)\\
&\geqslant&\Sl_{\sigma}(s;u)-\lambda d(s,t)+W_{\sigma}(s)-W_{\sigma}(t)\\
&\geqslant&\Sl_{\sigma}(s;C)-\lambda d(s,t)+W_{\sigma}(s)-W_{\sigma}(t).\\
\end{array}
\]
The first inequality is given by the triangle inequality and the second by Lemma~\ref{lem:subsetslack}. In general,
$W_{\sigma}(s)-W_{\sigma}(t)\geqslant -d(s,t)$, which implies $(a)$. If $t$ dominates $s$ then we have the stronger bound
$W_{\sigma}(s)-W_{\sigma}(t)= d(s,t)$.
\end{proof}

\bigskip

\section{The {\small CNN} problem}\label{sec: cnn problem}

A simple example shows that the standard work function algorithm $\WFA_{1}$ has unbounded competitive ratio for the \CNN problem: Take $(0,0)$
as the origin and consider the request sequence $(1,2),(2,2),(3,2),\dots (m,2)$ for arbitrary $m$. The optimal solution moves from $(0,0)$ to
$(0,2)$ but the work function algorithm follows the path $(0,0),(1,0),(2,0),\dots,(m,0)$. (There are no draws.) The competitive ratio for this
instance is $m/2$.

\smallskip\begin{theorem}\label{th:theorem1}
The generalized work function algorithm $\WFA_{\lambda}$ is constant competitive for the \CNN problem for any constant $\lambda$ with
$0<\lambda<1$.
\end{theorem}

All the lemmas of the previous section apply to metrical service systems in general. In this section, we restrict to the \CNN problem. It is
convenient to insist on writing $\MM$ for the metric space although we now have $\MM=\RR^2$. We  make a subtle distinction between the
\emph{request point} $(x',y')\in \RR^2$ and the corresponding \emph{request} as defined by the metrical service system: $r(x',y')=\{\{(x,y)\in
\RR^2\mid x=x'\text{ or }y=y'\}$.

\subsection{The potential function}\label{sec:potential cnn}
Our potential function is defined for any metrical service system on $\RR^2$ but we only use it for the $\CNN$ problem.

One of the ingredients is the set $\Boks(s_1,s_2)$ (see Figure~\ref{fig:boks}) defined as follows. Given points $x_1,x_2\in \RR$, we denote by
$[x_1,x_2]$  the interval between $x_1$ and $x_2$ (we allow $x_2<x_1$, i.e., $[x_2,x_1]=[x_1,x_2]$). Note that at this point we use the
restriction to the real line since this is not well-defined for a general metric space. (In Section~\ref{sec:2-server}, where the proof is
generalized to arbitrary metric spaces we shall start from this point.)

Given points $s_1=(x_1,y_1)\in \MM$ and $s_2=(x_2,y_2)\in \MM$ we denote the set of points in the rectangle spanned by these points by
\begin{equation*}\label{def:Boks} \Boks(s_1,s_2)=\{(x,y)\in \MM \mid x\in [x_1,x_2] \text{ and } y\in [y_1,y_2]\}.
\end{equation*}
Let $0<\alpha< 1/2$ and $0<\gamma<1$. We define the functions $\Ff_{\sigma}: \MM^3\rightarrow \RR$ and $\Gg_{\sigma}: \MM^3\rightarrow \RR$ as
\[
\begin{array}{rcl}
\Ff_{\sigma}(s_1,s_2,s_3)&=&W_{\sigma}(s_1)-\frac{1}{2}\Sl_{\sigma}(s_2;s_1)-\alpha \Sl_{\sigma}(s_3;\{s_1,s_2\})\\[3mm]
\Gg_{\sigma}(s_1,s_2,s_3)&=&W_{\sigma}(s_1)-\frac{1}{2}\Sl_{\sigma}(s_2;s_1)-\alpha \Sl_{\sigma}(s_3;\Boks(s_1,s_2)).\\
\end{array}
\]
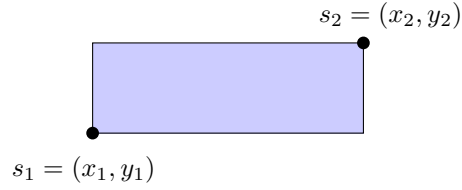
\begin{figure}
\center
\begin{tikzpicture}[scale=1.2]
\node at (0.9,0.6) {$s_1=(x_1,y_1)$};
\node at (4.3,2.3) {$s_2=(x_2,y_2)$};
\filldraw [fill=blue!20] (1,1) rectangle (4,2);
\fill (canvas cs:x=1cm,y=1cm) circle (2pt);
\fill (canvas cs:x=4cm,y=2cm) circle (2pt);
\end{tikzpicture}
\caption{The shaded area is $\Boks(s_1,s_2)$, used in the potential function.}
\label{fig:boks}
\end{figure}
The two functions only differ in the last term. The potential function $\Phi_{\sigma}$ is\\
\[\Phi_{\sigma}=(1-\gamma)\min_{s_1,s_2,s_3\in\MM}\Ff_{\sigma}(s_1,s_2,s_3)+\gamma\min_{s_1,s_2,s_3\in\MM}\Gg_{\sigma}(s_1,s_2,s_3).\]\\
The numbers $\alpha$ and $\gamma$  will depend only on $\lambda$ and we fix their precise values later.  It is good to mention here that the
proof works for any small enough values of $\alpha$ and $\gamma$. More precisely, the proof  works if we pick any $\alpha$ with
$0<\alpha\leqslant \alpha_0$ for some $\alpha_0$ depending on $\lambda$ and then pick any $\gamma$ with $0<\gamma\leqslant \gamma_0$ for some
$\gamma_0$ depending on $\lambda$ and $\alpha$.

\subsubsection*{Comprehensive notation} \label{subsubsec:notation}
To simplify the analysis we define one more function $\Hh_{\sigma}: \MM^2\rightarrow \RR$. It corresponds to the first two terms of
$\Ff_{\sigma}$ and $\Gg_{\sigma}$.
\[
\Hh_{\sigma}(s_1,s_2)=W_{\sigma}(s_1)-\frac{1}{2}\Sl_{\sigma}(s_2;s_1)=\frac{1}{2}W_{\sigma}(s_1)+\frac{1}{2}W_{\sigma}(s_2)-\frac{\lambda}{2}d(s_1,s_2).
\]
We can rewrite $\Ff_{\sigma}$ and $\Gg_{\sigma}$ as
\[
\begin{array}{rcl}
\Ff_{\sigma}(s_1,s_2,s_3)&=&\Hh_{\sigma}(s_1,s_2)-\alpha \Sl_{\sigma}(s_3;\{s_1,s_2\}),\\[3mm]
\Gg_{\sigma}(s_1,s_2,s_3)&=&\Hh_{\sigma}(s_1,s_2)-\alpha \Sl_{\sigma}(s_3;\Boks(s_1,s_2)).\\
\end{array}
\]
For a request sequence $\sigma$, we denote $\min_{s_1,s_2,s_3\in \MM}\Ff_{\sigma}(s_1,s_2,s_3)$ simply by $\min \Ff_{\sigma}$ and make a similar
simplification of notation for $\Gg$ and $\Hh$. A shorter notation for the potential function becomes
\[\Phi_{\sigma}=(1-\gamma)\min\Ff_{\sigma}+\gamma\min\Gg_{\sigma}.\]
Note that $\Hh_{\sigma}$ is symmetric in $s_1$ and $s_2$ and, consequently, also  $\Ff_{\sigma}$ and  $\Gg_{\sigma}$ are symmetric in $s_1$ and
$s_2$. This property is not essential but enhances the argumentation at some points.

\subsection{Properties of the potential function}\label{sec:prop potential}
In this section we list some properties of the potential function $\Phi_{\sigma}$ which hold for any metrical service system on $\MM=\RR^2$ and
arbitrary corresponding request sequence $\sigma$. In section~\ref{sec:prop cnn}, we restrict the analysis to the \CNN problem.

The functions $\Ff_{\sigma}$ and $\Gg_{\sigma}$ are constructed such that in the minimum all three points $s_1,s_2,s_3$ are on the last request,
at least if $\alpha$ is small enough. This is stated in Lemma~\ref{lem:minimum in r}. The next lemma is preliminary for this lemma and several
others.

\smallskip\begin{lemma}\label{lem:s dominates t}
Let  $t\in \MM$ dominate $s\in \MM$ (w.r.t. $\sigma$) and let  $\delta=d(s,t)$. Then, for any $s_1,s_2,s_3\in\MM$
\[
\begin{array}{llllcl}
(a)&\Ff_{\sigma}(s_1,s_2,s)&-&\Ff_{\sigma}(s_1,s_2,t)&\geqslant&
\delta\cdot\alpha(1-\lambda)\\

(b)&\Ff_{\sigma}(s_1,s,s_3)&-&\Ff_{\sigma}(s_1,t,s_3)&\geqslant& \delta\cdot\left(\frac{1}{2}(1-\lambda)-\alpha(1+\lambda)\right)\\

(c) &\Ff_{\sigma}(s,s_2,s_3)&-&\Ff_{\sigma}(t,s_2,s_3)&\geqslant& \delta\cdot\left(\frac{1}{2}(1-\lambda)-\alpha(1+\lambda)\right)\\

(d)&\Gg_{\sigma}(s_1,s_2,s)&-&\Gg_{\sigma}(s_1,s_2,t)&\geqslant&
\delta\cdot\alpha(1-\lambda).\\

(e)&\Gg_{\sigma}(s_1,s,s_3)&-&\Gg_{\sigma}(s_1,t,s_3)&\geqslant& \delta\cdot\left(\frac{1}{2}(1-\lambda)-\alpha(1+\lambda)\right)\\

(f)&\Gg_{\sigma}(s,s_2,s_3)&-&\Gg_{\sigma}(t,s_2,s_3)&\geqslant& \delta\cdot\left(\frac{1}{2}(1-\lambda)-\alpha(1+\lambda)\right).
\end{array}
\]
(Follows from Lemma~\ref{lem:slack change C1 to C2}, \ref{lem:slack change s to t}. Used in proof of Lemma~\ref{lem:minimum in r},~\ref{lem:phi
epsilon},~\ref{lem:domination},~\ref{lem:F increase}, and~\ref{lem:G increase}.)
\end{lemma}
\begin{proof}
Statement (a) follows  directly from Lemma~\ref{lem:slack change s to t}(b) with $C=\{s_1,s_2\}$:
\[\Ff_{\sigma}(s_1,s_2,s)-\Ff_{\sigma}(s_1,s_2,t)=\alpha\Sl_{\sigma}(t;\{s_1,s_2\})-\alpha\Sl_{\sigma}(s;\{s_1,s_2\})\geqslant \alpha(1-\lambda)\delta.\]
The same holds for (d) but now with $C=\Boks(s_1,s_2)$. By symmetry of $\Ff$ and $\Gg$ in their first two arguments, it only remains to prove
statements (b) and (e). We start with (b).
\begin{equation}\label{eq:0-lemma s to t}
\begin{array}{rcl}
&&\Ff_{\sigma}(s_1,s,s_3)-\Ff_{\sigma}(s_1,t,s_3)\\
&=&\frac{1}{2}\left(\Sl_{\sigma}(t;s_1)-\Sl_{\sigma}(s;s_1)\right)+\alpha \left(\Sl_{\sigma}(s_3;\{s_1,t\})- \Sl_{\sigma}(s_3;\{s_1,s\})\right).
\end{array}
\end{equation}
For the first part of~\eqref{eq:0-lemma s to t} we use Lemma~\ref{lem:slack change s to t}(b):
\begin{equation}\label{eq:1-lemma s to t}
\Sl_{\sigma}(t;s_1)- \Sl_{\sigma}(s;s_1)\geqslant(1-\lambda)\delta.
\end{equation}
For the second part we apply Lemma~\ref{lem:slack change C1 to C2} with $C_1=\{s_1,t\}$ and $C_2=\{s_1,s\}$. The condition of
Lemma~\ref{lem:slack change C1 to C2} is satisfied for $\delta=d(s,t)$. We have
\begin{equation}\label{eq:2-lemma s to t}
\Sl_{\sigma}(s_3;\{s_1,t\})-\Sl_{\sigma}(s_3;\{s_1,s\})\geqslant -(1+\lambda)\delta.
\end{equation}
Combining~\eqref{eq:1-lemma s to t} and~\eqref{eq:2-lemma s to t} we get (b).
The proof of (e) is similar. We apply~\eqref{eq:1-lemma s to t} and Lemma~\ref{lem:slack change C1 to C2} with $C_1=\Boks(s_1,t)$ and $C_2=\Boks(s_1,s)$. The condition of
Lemma~\ref{lem:slack change C1 to C2} is satisfied for $\delta=d(s,t)$.
\[
\begin{array}{rcl}
&&\Gg_{\sigma}(s_1,s,s_3)-\Gg_{\sigma}(s_1,t,s_3)\\
&=&\frac{1}{2}\left(\Sl_{\sigma}(t;s_1)-\Sl_{\sigma}(s;s_1)\right)+\alpha \left(\Sl_{\sigma}(s_3;\Boks(s_1,t))- \Sl_{\sigma}(s_3;\Boks(s_1,s))\right)\\
&\geqslant &\frac{1}{2} (1-\lambda)\delta- \alpha(1+\lambda)\delta.\end{array}
\]
\end{proof}\bigskip

Note that all the right hand sides in  Lemma~\ref{lem:s dominates t} are strictly positive if  $0<\alpha <(1-\lambda)/(2(1+\lambda))$. We assume
this from now on.

\smallskip\begin{lemma}\label{lem:minimum in r}
If $\Ff_{\sigma r}$ or $\Gg_{\sigma r}$ is minimized in $(s_1,s_2,s_3)$, then $s_1,s_2,s_3\in r$. (Follows from Lemma~\ref{lem:s dominates t}.
Used in proof of Lemma~\ref{lem:F increase} and~\ref{lem:G increase}.)
\end{lemma}
\begin{proof}
Any point is dominated by a point in the last request. Take any triple of points in $\MM$. If at least one of the points is not in $r$, then
Lemma~\ref{lem:s dominates t}  tells us that we can replace it  by a point of the last request, $r$, such that the values of $\Ff_{\sigma}$ and
$\Gg_{\sigma}$ become strictly smaller.  \end{proof}\bigskip

\smallskip\begin{lemma}\label{lem:F<=G<=H}
$\min \Ff_{\sigma}\leqslant \min\Gg_{\sigma}\leqslant\min\Hh_{\sigma}$. (Follows from Lemma~\ref{lem:subsetslack} and~\ref{lem:minimum slack}.
Used in proof of Lemma~\ref{lem:minG''ge minG'}.)
\end{lemma}
\begin{proof}
For any $s_1,s_2\in\MM$ there is a point $s_3$ such that $\Sl_{\sigma}(s_3;\Boks(s_1,s_2))=0$. (See Lemma~\ref{lem:minimum slack}.) Hence, $\min
\Gg_{\sigma} \leqslant \min \Hh_{\sigma}$.

 For any $s_1,s_2,s_3\in\MM$ we have $\Sl_{\sigma}(s_3;\{s_1,s_2\})\geqslant \Sl_{\sigma}(s_3;\Boks(s_1,s_2))$
since $\{s_1,s_2\}\subseteq \Boks(s_1,s_2)$. (See Lemma~\ref{lem:subsetslack}.) Therefore,  $\min \Ff_{\sigma}\leqslant \min\Gg_{\sigma}$.
\end{proof}\bigskip

The two inequalities of Lemma~\ref{lem:F<=G<=H} are only strict if the three points for which the minimum of $\Ff_{\sigma}$ or $\Gg_{\sigma}$ is
attained are in a way different enough. For example, the next lemma implies that if the minimum of $\Ff_{\sigma}$ is attained for
$(s_1,s_2,s_3)$ but they are not all different, then both inequalities are equalities. For $\Gg_{\sigma}$ a stronger property holds. If
$s_1,s_2,s_3$ are all on a line then the second inequality is an equality.

\smallskip\begin{lemma}\label{lem:cardinality=2}
If $\{s_1,s_2,s_3\}$ has cardinality 1 or 2, then $\Hh_{\sigma}(u_1,u_2)\leqslant \Ff_{\sigma}(s_1,s_2,s_3)$ for some $u_1,u_2\in\{s_1,s_2,s_3\}$.
(Used in proof of Lemma~\ref{lem:F increase} and~\ref{lem:minG''ge minG'}.)
\end{lemma}
\begin{proof}
If $\Sl_{\sigma}(s_3;\{s_1,s_2\})\leqslant 0$, then  $\Hh_{\sigma}(s_1,s_2)\leqslant \Ff_{\sigma}(s_1,s_2,s_3)$. So assume the opposite:
\begin{equation}\label{eq:slack>0-1}
\Sl_{\sigma}(s_3;\{s_1,s_2\})> 0.
\end{equation}
We cannot have $s_1=s_3$ or $s_2=s_3$, since this contradicts~\eqref{eq:slack>0-1}. Hence, we must have $s_1=s_2$, which implies
$\Sl_{\sigma}(s_2;s_1)=0$.
\[
\begin{array}{rcl}
\Ff_{\sigma}(s_1,s_2,s_3)&=&W_{\sigma}(s_1)-\frac{1}{2} \Sl_{\sigma}(s_2;s_1)-\alpha \Sl_{\sigma}(s_3;\{s_1,s_2\})\\
&=&W_{\sigma}(s_1)-\alpha \Sl_{\sigma}(s_3;\{s_1,s_2\})\\
&>&W_{\sigma}(s_1)-\frac{1}{2} \Sl_{\sigma}(s_3;\{s_1,s_2\})\\
&=&W_{\sigma}(s_1)-\frac{1}{2} \Sl_{\sigma}(s_3;s_1)\\
&=&\Hh_{\sigma}(s_1,s_3).
\end{array}
\]
For the inequality we used~\eqref{eq:slack>0-1} and $\alpha<1/2$.
\end{proof}\bigskip

Lemma~\ref{lem:cardinality=2} applies also to $\Gg_{\sigma}$ instead of $\Ff_{\sigma}$ but we shall not use this. In addition, $\Gg_{\sigma}$
has the following property.

\smallskip\begin{lemma}\label{lem:convex containment}
If  $s_1,s_2,s_3\in \MM$ all have the same $x$-coordinate or the same $y$-coordinate, then, $\Hh_{\sigma}(u_1,u_2)\leqslant
\Gg_{\sigma}(s_1,s_2,s_3)$ for some $u_1,u_2\in \{s_1,s_2,s_3\}$. (Follows from Lemma~\ref{lem:subsetslack}, \ref{lem:transitive slack}. Used in
proof of Lemma~\ref{lem:G increase}.)
\end{lemma}
\begin{proof}
We shall prove something stronger than we need as this hardly changes the proof: If one of the three points is contained in $\Boks(\cdot,\cdot)$
defined by the other two points, then $\Hh_{\sigma}(u_1,u_2)\leqslant \Gg_{\sigma}(s_1,s_2,s_3)$ for some $u_1,u_2\in \{s_1,s_2,s_3\}$. The lemma
is a special case of this.

The proof is similar to that of Lemma~\ref{lem:cardinality=2}. If $\Sl_{\sigma}(s_3;\Boks(s_1,s_2))\leqslant 0$, then
$\Hh_{\sigma}(s_1,s_2)\leqslant \Gg_{\sigma}(s_1,s_2,s_3)$ and we are done. So assume the opposite:
\begin{equation}\label{eq:slack>0-2}
\Sl_{\sigma}(s_3;\Boks(s_1,s_2))> 0.
\end{equation}
Now assume $s_1\in\Boks(s_2,s_3)$ or $s_2\in\Boks(s_1,s_3)$ or $s_3\in\Boks(s_1,s_2)$. We cannot have the latter since that
contradicts~\eqref{eq:slack>0-2}. Hence, either $s_1$ or $s_2$ is contained in $\Boks(\cdot,\cdot)$ defined by the other two points. By symmetry
of $\Hh$ and $\Gg$ in their first two arguments, we may assume the latter is true. Hence, $d(s_1,s_2)+d(s_2,s_3)=d(s_1,s_3)$. By
Lemma~\ref{lem:transitive slack},
\begin{equation}\label{eq:slack equality}
\Sl_{\sigma}(s_3;s_1)=\Sl_{\sigma}(s_3;s_2)+\Sl_{\sigma}(s_2;s_1).
\end{equation}
For the first inequality below we use Lemma~\ref{lem:subsetslack} and for the second we use  $\alpha< 1/2$ and~\eqref{eq:slack>0-2}.
\[
\begin{array}{rcl}
\Gg_{\sigma}(s_1,s_2,s_3)&=& W_{\sigma}(s_1)-\frac{1}{2} \Sl_{\sigma}(s_2;s_1)-\alpha \Sl_{\sigma}(s_3;\Boks(s_1,s_2))\\
&\geqslant& W_{\sigma}(s_1)-\frac{1}{2}\Sl_{\sigma}(s_2;s_1)-\alpha \Sl_{\sigma}(s_3;s_2)\\
&> & W_{\sigma}(s_1)-\frac{1}{2}\Sl_{\sigma}(s_2;s_1)- \frac{1}{2} \Sl_{\sigma}(s_3;s_2)\\
&=& W_{\sigma}(s_1)-\frac{1}{2} \Sl_{\sigma}(s_3;s_1)\\
&=&\Hh_{\sigma}(s_1,s_3).
\end{array}
\]
\end{proof}\bigskip

Initially, the potential function is zero and in general it is upper bounded by the optimal value of the given sequence. This is stated in the
next two lemmas. Let $\epsilon$ be the empty request sequence.
\smallskip\begin{lemma}\label{lem:phi epsilon}
$\Phi_{\epsilon}=0$. (Follows from Lemma~\ref{lem:s dominates t}. Used in proof of Theorem~\ref{th:theorem1}.)
\end{lemma}
\begin{proof}
Any point $s$ is dominated by the origin $\mathcal{O}$, w.r.t. the empty sequence. By Lemma~\ref{lem:s dominates t}, we see that
$\min\Ff_{\epsilon}=\Ff_{\epsilon}(\mathcal{O},\mathcal{O},\mathcal{O})=0$ and
$\min\Gg_{\epsilon}=\Gg_{\epsilon}(\mathcal{O},\mathcal{O},\mathcal{O})=0$.
\end{proof}\bigskip

\smallskip\begin{lemma}\label{lem:phi sigma}
$\Phi_{\roo}\leqslant \Opt_{\roo}$, for any sequence $\roo$. (Used in proof of Theorem~\ref{th:theorem1}.)\end{lemma}
\begin{proof}
Let $q$ be the endpoint of an optimal solution for $\roo$. Then $W_{\roo}(q)=\Opt_{\roo}$ and $\Ff_{\roo}(q,q,q)=\Gg_{\roo}(q,q,q)=W_{\roo}(q)$.
Hence, $\min \Ff_{\roo}\leqslant W_{\roo}(q)=\Opt_{\roo}$ and $\min \Gg_{\roo}\leqslant W_{\roo}(q)=\Opt_{\roo}$.
\[\Phi_{\roo}=(1-\gamma)\min\Ff_{\roo}+\gamma\min\Gg_{\roo}\leqslant (1-\gamma)\Opt_{\roo}+\gamma\Opt_{\roo}=\Opt_{\roo}.\]

\end{proof}\bigskip

\subsection{Properties of the work function}\label{sec:prop cnn}
Each metrical service system has its own specific properties of its work function. For example, Koutsoupias and Papadimitriou show a
quasi-convexity property of the work function for the $k$-server problem~\cite{KouPap95A}. A good understanding of the \CNN work function is
lacking but the two simple properties we show in this section are enough to prove constant competitiveness. These properties hold for the
generalized two server problem as well.

Let $\sigma'$ be an arbitrary request sequence for the \CNN problem and let $r'=r(x{'},y{'})$ be the last request in $\sigma'$.

\smallskip\begin{lemma}\label{lem:s in r dominted by}
Any $(x,y)\in \MM$  is dominated w.r.t. $\sigma'$ by $(x',y)$ or by  $(x,y')$. (Used in proof of Lemma~\ref{lem:domination} and~\ref{lem:G
increase}.)
\end{lemma}
\begin{proof}
Any point is dominated by a point of the last request. Therefore, $(x,y)$  is dominated by $(x',\hat{y})$ or by  $(\hat{x},y')$ for some
$\hat{y}\in\YY$ or $\hat{x}\in\XX$. In general, if $s$ is dominated by $t$, then $s$ is dominated by any point on the shortest path between $s$
and $t$. Now, note that $(x',y)$ is on the shortest path between  $(x,y)$ and $(x',\hat{y})$, and that $(x,y')$ is on the shortest path between
$(x,y)$ and $(\hat{x},y')$.  \end{proof}\bigskip
\begin{figure}
\center
\begin{tikzpicture}[scale=1]
\draw (0,1) node(y1L)  { };
\draw (10,1) node(y1R)  { };
\draw (0,2) node(y2L)  { };
\draw (10,2) node(y2R)  { };
\draw (2,0) node(x1B)  { };
\draw (8,0) node(x2B)  { };
\draw (2,3) node(x1T)  { };
\draw (8,3) node(x2T)  { };
\fill (canvas cs:x=2cm,y=1cm) circle (3pt);
\fill (canvas cs:x=8cm,y=2cm) circle (3pt);
\draw (x1B)--(x1T);
\draw [line width=1.5pt] (x2B)--(x2T);
\draw (y1L)--(y1R);
\draw [line width=1.5pt] (y2L)--(y2R);
\draw[<->] [dashed]  (1,1.1) -- (1,1.9);
\draw[<->] [dashed]  (2.1,2.8) -- (7.9,2.8);
\node at (5,3) {$\dd x$};
\node at (0.7,1.5) {$\dd y$};
\node at (1.3,0.6) {$(x',y')$};
\node at (8.7,2.4) {$(x'',y'')$};
\end{tikzpicture}
\caption{Two subsequent requests $r'=r(x',y')$ and $r''=r(x'',y'')$. We assume $\dd x\geqslant \dd y$.}
\label{fig:cnn}
\end{figure}
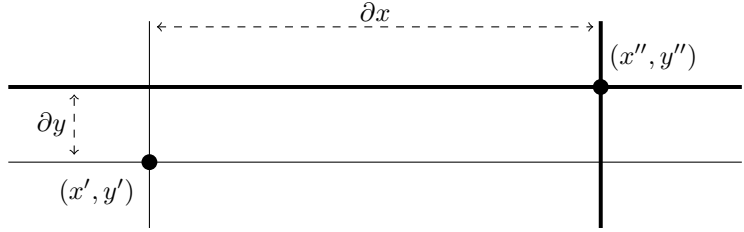

Let $\sigma'$ be followed by request $r''=r(x'',y'')$ and denote the extended sequence by $\sigma'' =\sigma' r''$. To simplify notation we denote
$d_{\XX}(x{'},x{''})=|x'-x''|$ by $\dd x$ and do the same for $y$. (See Figure~\ref{fig:cnn}.) From now on we assume without loss of generality
that
\[\dd x\geqslant \dd y.\]
Remember the definition of extended cost. From Equation~\eqref{eq:nabla for r rewritten} we know that
\[
\nnabla_{r''}(W_{\sigma'})=\max_{s\in \MM} \Sl_{\sigma'}(s;r'').
\]
Since this will be the only extended cost that we consider in this proof, we denote it simply by $\nnabla$. Further, let $\snabla\in \MM$ be a
point where the maximum is attained, i.e.,
\begin{equation}
\label{eq:nabla snabla}
\nnabla=\nnabla_{r''}(W_{\sigma'})=\Sl_{\sigma'}(\snabla;r'').
\end{equation}
Point $\snabla$ will be used in Lemma~\ref{lem:F increase} and~\ref{lem:G increase}.

\smallskip\begin{lemma}\label{lem:nabla dx}
$\nnabla\leqslant (1+\lambda)\dd x$. (Follows from Lemma~\ref{lem:slack change s to t}. Used in proof of Lemma~\ref{lem:G increase}.)
\end{lemma}
\begin{proof}Any $s\in \MM$ is dominated by a point in $r'$ w.r.t. $\sigma'$. Hence, by Lemma~\ref{lem:slack change s to t}(b), we may restrict to $r'$, i.e.,
$\nnabla=\max_{s\in \MM} \Sl_{\sigma'}(s;r'')=\max_{s\in r'} \Sl_{\sigma'}(s;r'')$. For any point $s$ in $r'$, there is a point in $r''$ at
distance at most $\dd x$, implying (using~\eqref{eq:Slack 1+lambda}) $\Sl_{\sigma'}(s;r'')\leqslant (1+\lambda)\dd x$ for any point $s$ in $r'$.
\end{proof}\bigskip

\subsection{The potential applied to \small CNN}\label{sec:pot applied}
In this section, we apply our potential function to the \CNN problem. Lemmas~\ref{lem:F increase},~\ref{lem:G increase} and~\ref{lem:minG''ge
minG'} state how $\min\Ff$ and $\min\Gg$ increase when a new request $r''$ arrives, i.e., when going from $\sigma'$ to $\sigma''$. Then,
Lemma~\ref{lem:nabla increase} combines these results and gives a lower bound on the increase of the potential function in terms of the extended
cost. The proof of Theorem~\ref{th:theorem1} is then straightforward.

The following lemma is given without proof as it is easy to check by looking at the definitions.
\smallskip\begin{lemma}\label{lem:G Lipschitz}
Consider request sequence $\sigma'$ as fixed and consider the next request $r''=r(x'',y''$) as a variable. Then $\min\Gg_{\sigma' r''}$ and
$\nnabla$ are Lipschitz continuous in both $x''$ and $y''$. (Used in proof of Lemma~\ref{lem:G increase}.)
\end{lemma}\smallskip

Lemma~\ref{lem:G Lipschitz} is true as well for $\Ff$ and $\Hh$ but we do not need that. We shall use the following easy property several times.
\begin{equation}\label{eq:w'=w''} W_{\sigma'}(s)=W_{\sigma''}(s), \text{ for any }s\in r''.
\end{equation}

Any $s\in \MM$ is dominated w.r.t. $\sigma'$ by a point in $r'$ and Lemma~\ref{lem:s in r dominted by} gives two candidate points. The next lemma
reduces this to one candidate in certain cases.

\smallskip\begin{lemma}\label{lem:domination}
Assume that $\Ff_{\sigma''}$ or $\Gg_{\sigma''}$ is minimized in $(s_1,s_2,s_3)$. Then, the following is true for any $i\in \{1,2,3\}$.
\begin{enumerate}
\item If $s_i=(x'',y)$ for some $y\neq y'$, then $(x',y)$ dominates $s_i$ w.r.t. ${\sigma'}$.
\item If $s_i=(x,y'')$ for some $x\neq x'$, then $(x,y')$ dominates $s_i$ w.r.t. ${\sigma'}$.
\end{enumerate}
(Follows from Lemma~\ref{lem:s dominates t}, \ref{lem:s in r dominted by}. Used in proof of Lemma~\ref{lem:F increase} and~\ref{lem:G increase}.)
\end{lemma}
\begin{proof}
We only prove the first, since the second follows by symmetry. By Lemma~\ref{lem:s in r dominted by}, point $s_i=(x'',y)$ is dominated w.r.t.
${\sigma'}$ by $(x'',y')$ or by $(x',y)$. Suppose the first is true. Then, using~\eqref{eq:w'=w''}, $s_i$ is dominated by this point w.r.t.
${\sigma''}$ as well. In that case Lemma~\ref{lem:s dominates t} implies that $\Ff_{\sigma''}$ and $\Gg_{\sigma''}$ are strictly reduced by
replacing $s_i$ by $(x'',y')$. This contradicts the assumption of minimality. Thus, $s_i$ is dominated by $(x',y)$ w.r.t. ${\sigma'}$.
\end{proof}\bigskip

Equation~\eqref{eq:w'=w''} implies that if $s_1,s_2,s_3\in r''$  then,
\[\Hh_{\sigma'}(s_1,s_2)=\Hh_{\sigma''}(s_1,s_2)\text{, and }\Ff_{\sigma'}(s_1,s_2,s_3)=\Ff_{\sigma''}(s_1,s_2,s_3).\] (This is not true for $\Gg$.) These two easy equalities will be used several times without reference in the following lemmas. 
From now, we let
\begin{equation}\label{eq:set alpha} \alpha\leqslant \frac{1-\lambda}{12+4\lambda}.\end{equation}
We use this bound for Lemma~\ref{lem:F increase} and Lemma~\ref{lem:G increase}, although for Lemma~\ref{lem:F increase} we could do with a
weaker bound.

\smallskip\begin{lemma}\label{lem:F increase}
Let $\Ff_{\sigma''}$ be minimized in  $(s_1,s_2,s_3)$. There are constants $c_1,c_2>0$ (depending on $\lambda$) such that,\\[2mm]
(Case A)\  if the cardinality of $\{s_1,s_2,s_3\}$ is 1 or 2 then
\[\min\Ff_{\sigma''}-\min\Ff_{\sigma'}\geqslant c_1 \nnabla,\text{ and }\]
(Case B)\  if the cardinality of $\{s_1,s_2,s_3\}$ is 3 then,
\[\min\Ff_{\sigma''}-\min\Ff_{\sigma'}\geqslant c_2 \dd y.\]
(Follows from Lemma~\ref{lem:s dominates t}, \ref{lem:minimum in r}, \ref{lem:cardinality=2}, \ref{lem:domination}. Used in proof of
Lemma~\ref{lem:nabla increase}.)
\end{lemma}
\begin{proof} Lemma~\ref{lem:minimum in r} tells us that $s_1,s_2,s_3\in r''$. We use this in both cases.\\
\emph{Case A:} By Lemma~\ref{lem:cardinality=2}, there are points $u_1,u_2\in\{s_1,s_2,s_3\}$ such that  $\Hh_{\sigma''}(u_1,u_2)\leqslant
\Ff_{\sigma''}(s_1,s_3,s_3)$. Remember the definition of $\snabla$ in~\eqref{eq:nabla snabla}.
\[
\begin{array}{rcl}
\min\Ff_{\sigma''}&=&\Ff_{\sigma''}(s_1,s_2,s_3)\\
&\geqslant&\Hh_{\sigma''}(u_1,u_2)\\
&=&\Hh_{\sigma'}(u_1,u_2)\\
&=&\Ff_{\sigma'}(u_1,u_2,\snabla)+\alpha \Sl_{\sigma'}(\snabla;\{u_1,u_2\})\\
&\geqslant&\min\Ff_{\sigma'}+\alpha \Sl_{\sigma'}(\snabla;\{u_1,u_2\})\\
&\geqslant&\min\Ff_{\sigma'}+\alpha \Sl_{\sigma'}(\snabla;r'')\\
&=&\min\Ff_{\sigma'}+\alpha \nnabla.
\end{array}
\]
\emph{Case B}:  Since the three points are different, at least one of the these points differs from both $(x'',y')$ and $(x',y'')$. By
Lemma~\ref{lem:domination}, this point is dominated w.r.t. $\sigma'$ by a point at distance $\dd x$ or $\dd y$. Now we use Lemma~\ref{lem:s
dominates t} with $\delta=\dd y$. Note that, by our bound on $\alpha$, all righthand sides in Lemma~\ref{lem:s dominates t} are at least
$\alpha(1-\lambda)\dd y$.
\[
\begin{array}{rcl}
\min\Ff_{\sigma''}&=& \Ff_{\sigma''}(s_1,s_2,s_3)\\
&=& \Ff_{\sigma'}(s_1,s_2,s_3)\\
&\geqslant &\min\Ff_{\sigma'}+\alpha(1-\lambda)\dd y.
\end{array}
\]
\end{proof}\bigskip

\smallskip\begin{lemma}\label{lem:G increase}
$\min\Gg_{\sigma''}-\min\Gg_{\sigma'}\geqslant c_3\nnabla-c_4\dd y$,  for some constants $c_3,c_4>0$ depending on $\lambda$. (Follows from
Lemma~\ref{lem:subsetslack}, \ref{lem:s dominates t}, \ref{lem:minimum in r}, \ref{lem:convex containment}, \ref{lem:nabla dx}, \ref{lem:G
Lipschitz}, \ref{lem:domination}. Used in proof of Lemma~\ref{lem:nabla increase}.)
\end{lemma}
\begin{proof}
By Lemma~\ref{lem:G Lipschitz}, it is enough to prove that $\min\Gg_{\sigma''}-\min\Gg_{\sigma'}\geqslant c_3\nnabla$ under the assumption that
$y'=y''$. So we assume $y'=y''$ and denote both by $\yy$.

Let $\Gg_{\sigma''}$ be minimized for $(s_1,s_2,s_3)$ and let $s_i=(x_i,y_i)$, for $i\in \{1,2,3\}$. We make the following partition of possible
cases.

\begin{tabular}{llclcl}
\emph{Case 1: }& $\ y_1=\yy$ &and & $y_2= \yy$ &and & $y_3= \yy$,\\
\emph{Case 2: }& $\ y_1=\yy$ &and & $y_2= \yy$ &and & $y_3\neq \yy$,\\
\emph{Case 3: }& $\ y_1\neq \yy$ &or & $y_2\neq \yy$. && \\
\end{tabular}\smallskip

By Lemma~\ref{lem:minimum in r}, we have $s_1,s_2,s_3\in r''$. We shall use this property several times here. For example, if $y_i\neq \yy$ then
$x_i=x''$. \bigskip

\noindent\emph{Case 1: } We apply Lemma~\ref{lem:convex containment}: $\Hh_{\sigma''}(u_1,u_2)\leqslant \Gg_{\sigma''}(s_1,s_2,s_3)$ for some
$u_1,u_2\in \{s_1,s_2,s_3\}$.
\[
\begin{array}{rcl}
\min\Gg_{\sigma''}&=&\Gg_{\sigma''}(s_1,s_2,s_3)\\
&\geqslant &\Hh_{\sigma''}(u_1,u_2)\\
&= &\Hh_{\sigma'}(u_1,u_2)\\
&= &\Gg_{\sigma'}(u_1,u_2,\snabla)+\alpha \Sl_{\sigma'}(\snabla;\Boks(u_1,u_2))\\
&\geqslant &\min\Gg_{\sigma'}+\alpha \Sl_{\sigma'}(\snabla;\Boks(u_1,u_2))\\
&\geqslant &\min\Gg_{\sigma'}+\alpha \Sl_{\sigma'}(\snabla;r'')\\
&= &\min\Gg_{\sigma'}+\alpha \nnabla.\\
\end{array}
\]
The last inequality follows from $\Boks(u_1,u_2)\subset r''$ and Lemma~\ref{lem:subsetslack}.\bigskip

\noindent\emph{Case 2: } Since $\Boks(s_1,s_2)\subset r''$ and $s_1,s_2,s_3\in r''$ we have
\[\Gg_{\sigma''}(s_1,s_2,s_3)=\Gg_{\sigma'}(s_1,s_2,s_3).\]
By Lemma~\ref{lem:domination}, point $s_3=(x'',y_3)$ is dominated by point $t=(x',y_3)$ with respect to $\sigma'$. Now we apply Lemma~\ref{lem:s
dominates t}-$d$.
\begin{equation*}\label{eq:Case 2}
\begin{array}{rcl}
\min\Gg_{\sigma''}&=&\Gg_{\sigma''}(s_1,s_2,s_3)\\
&=&\Gg_{\sigma'}(s_1,s_2,s_3)\\
&\geqslant&\Gg_{\sigma'}(s_1,s_2,t)+\alpha(1-\lambda)\dd x\\
&\geqslant&\min \Gg_{\sigma'}+\alpha(1-\lambda)\dd x\\
&\geqslant&\min \Gg_{\sigma'}+\alpha(1-\lambda)\nnabla / (1+\lambda).\\
\end{array}
\end{equation*}
The last inequality is given by Lemma~\ref{lem:nabla dx}.

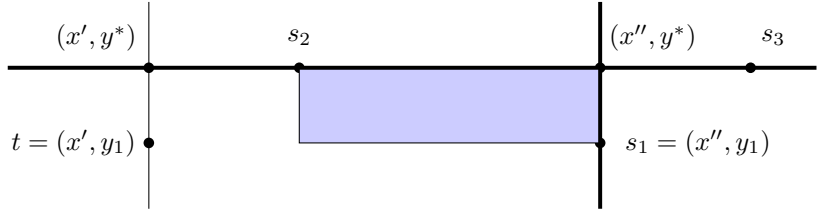
\begin{figure}
\center
\begin{tikzpicture}[scale=1]
\draw (0,1) node(y1L)  { };
\draw (10,1) node(y1R)  { };
\draw (0,2) node(y2L)  { };
\draw (11,2) node(y2R)  { };
\draw (2,0) node(x1B)  { };
\draw (8,0) node(x2B)  { };
\draw (2,3) node(x1T)  { };
\draw (8,3) node(x2T)  { };
\fill (canvas cs:x=2cm,y=1cm) circle (2pt);
\fill (canvas cs:x=8cm,y=2cm) circle (2pt);
\fill (canvas cs:x=8cm,y=1cm) circle (2pt);
\fill (canvas cs:x=2cm,y=2cm) circle (2pt);
\fill (canvas cs:x=4cm,y=2cm) circle (2pt);
\fill (canvas cs:x=10cm,y=2cm) circle (2pt);
\filldraw [fill=blue!20] (4,1) rectangle (8,2);
\draw (x1B)--(x1T);
\draw [line width=1.5pt] (x2B)--(x2T);
\draw [line width=1.5pt] (y2L)--(y2R);
\node at (1.3,2.4) {$(x',\yy)$};
\node at (4,2.4) {$s_2$};
\node at (1,1) {$t=(x',y_1)$};
\node at (9.3,1) {$s_1=(x'',y_1)$};
\node at (8.7,2.4) {$(x'',\yy)$};
\node at (10.3,2.4) {$s_3$};
\end{tikzpicture}
\caption{Case 3 of Lemma~\ref{lem:G increase}: $y_1\neq \yy$. The shaded area is $\Boks(s_1,s_2)$.}
\label{fig:case3}
\end{figure}

\noindent\emph{Case 3: } Unlike the previous two cases, we may now have $\Boks(s_1,s_2)\nsubseteq r''$ which makes the proof slightly more
complicated (see Figure~\ref{fig:case3}). By symmetry, we may assume that $y_1\neq \yy$. This implies $x_1=x''$ and point $s_1=(x'',y_1)$ is
dominated by point $t=(x',y_1)$ with respect to $\sigma'$. We now apply Lemma~\ref{lem:s dominates t}$(f)$.
\begin{equation}\label{eq:case3 1}
\begin{array}{rcl}
\Gg_{\sigma'}(s_1,s_2,s_3)&\geqslant&\Gg_{\sigma'}(t,s_2,s_3)+\left(\frac{1}{2}(1-\lambda)-\alpha(1+\lambda)\right)\dd x\\
&\geqslant &\min\Gg_{\sigma'}+\left(\frac{1}{2}(1-\lambda)-\alpha(1+\lambda)\right)\dd x.
\end{array}
\end{equation}
It remains to bound $\min\Gg_{\sigma''}-\Gg_{\sigma'}(s_1,s_2,s_3)$. We have~\footnote{A more careful analysis gives a bound
$-(1+\lambda)\alpha\dd x$ instead of $-2\alpha\dd x$.}
 \begin{equation}\label{eq:case3 2}
\begin{array}{rcl}
&&\min \Gg_{\sigma''}-\Gg_{\sigma'}(s_1,s_2,s_3)\\
&=&\Gg_{\sigma''}(s_1,s_2,s_3)-\Gg_{\sigma'}(s_1,s_2,s_3)\\
&=&\alpha \Sl_{\sigma'}(s_3;\Boks(s_1,s_2))-\alpha\Sl_{\sigma''}(s_3;\Boks(s_1,s_2))\\
&\geqslant &-2\alpha\dd x. \end{array}
\end{equation}
The last inequality follows from $W_{\sigma'}(s_3)=W_{\sigma''}(s_3)$ and from $W_{\sigma''}(s)-W_{\sigma'}(s)\leqslant 2\dd x$ for any point
$s\in \MM$ (and $s\in \Boks(s_1,s_2)$ in particular). Below we use, subsequently,~\eqref{eq:case3 2},~\eqref{eq:case3 1},~\eqref{eq:set alpha}
and Lemma~\ref{lem:nabla dx}.
\[
\begin{array}{rcl}
&&\min\Gg_{\sigma''}\\
&=&\Gg_{\sigma''}(s_1,s_2,s_3)\\
&\geqslant &\Gg_{\sigma'}(s_1,s_2,s_3)-2\alpha \dd x\\
&\geqslant &\min\Gg_{\sigma'}+\left(\frac{1}{2}(1-\lambda)-\alpha(1+\lambda)\right)\dd x-2\alpha\dd x\\
&=&\min\Gg_{\sigma'}+\left(\frac{1}{2}(1-\lambda)-\alpha(3+\lambda)\right)\dd x\\
&\geqslant &\min\Gg_{\sigma'}+\left(\frac{1}{2}(1-\lambda)-\frac{1}{4}(1-\lambda)\right)\dd x\\
&= &\min\Gg_{\sigma'}+\frac{1}{4}(1-\lambda)\dd x\\
&\geqslant &\min\Gg_{\sigma'}+\frac{1}{4}(1-\lambda)\nnabla/(1+\lambda).
\end{array}
\]
This completes the proof of the last case.
\end{proof}\bigskip

In Lemma~\ref{lem:nabla increase}, we combine Lemmas~\ref{lem:F increase} and \ref{lem:G increase} and distinguish the same two cases A and B as
we did in Lemma~\ref{lem:F increase}. Lemma~\ref{lem:G increase} will be used only for Case B, although it holds in general. For Case A, we need
the following different bound.

\smallskip\begin{lemma}\label{lem:minG''ge minG'}
Let $\Ff_{\sigma''}$ be minimized in  $(s_1,s_2,s_3)$. If the cardinality of $(s_1,s_2,s_3)$ is 1 or 2 then $\min\Gg_{\sigma''}\geqslant
\min\Gg_{\sigma'}$.(Follows from Lemma~\ref{lem:F<=G<=H}, \ref{lem:cardinality=2}. Used in proof of Lemma~\ref{lem:nabla increase}.)
\end{lemma}
\begin{proof}
By Lemma~\ref{lem:cardinality=2}, there are points $u_1,u_2\in \{s_1,s_2,s_3\}$ such that
$\min\Hh_{\sigma''}\leqslant\Hh_{\sigma''}(u_1,u_2)\leqslant \Ff_{\sigma''}(s_1,s_2,s_3)=\min\Ff_{\sigma''}$. In Lemma~\ref{lem:F<=G<=H}, the
inequalities are the other way around. Hence,
\[
\min
\Ff_{\sigma''}=\min \Gg_{\sigma''}=\min \Hh_{\sigma''}.\]
Further, note that $\min\Hh_{\sigma''}\geqslant \min\Hh_{\sigma'}$ (since, by defintion of $\Hh$, we have that  $\Hh_{\sigma''}(t_1,t_2)\geqslant \Hh_{\sigma'}(t_1,t_2)$ for any pair of points $t_1,t_2$). We conclude that
\[
\min\Gg_{\sigma''}=\min\Hh_{\sigma''}\geqslant \min\Hh_{\sigma'}\geqslant\min\Gg_{\sigma'},
\]
where the last inequality follows again from Lemma~\ref{lem:F<=G<=H}. 
\end{proof}\bigskip

\smallskip\begin{lemma}\label{lem:nabla increase}
$\Phi_{\sigma''}- \Phi_{\sigma'}\geqslant c_5 \nnabla$ for some constant $c_5>0$, depending on $\lambda$. (Follows from Lemma~\ref{lem:F
increase}, \ref{lem:G increase}. Used in proof of Theorem~\ref{th:theorem1}.)
\end{lemma}
\begin{proof}
\begin{eqnarray*}\label{eq:case a}
\Phi_{\sigma''}-\Phi_{\sigma'}&=&(1-\gamma)\left(\min \Ff_{\sigma''}-\min \Ff_{\sigma'}\right)+\gamma\left(\min \Gg_{\sigma''}-\min \Gg_{\sigma'}\right)\nonumber
\end{eqnarray*}
Let $\Ff_{\sigma''}$ be minimized in  $(s_1,s_2,s_3)$. We distinguish between the same two cases as in Lemma~\ref{lem:F increase}:\\[2mm]
\emph{Case A:} The cardinality of $\{s_1,s_2,s_3\}$ is 1 or 2.\\
\emph{Case B:} The cardinality of $\{s_1,s_2,s_3\}$ is 3.\\

\noindent\emph{Case A:} By Lemma~\ref{lem:F increase}, $\min\Ff_{\sigma''}-\min\Ff_{\sigma'}\geqslant c_1\nnabla$, for some constant $c_1>0$ and
by Lemma~\ref{lem:minG''ge minG'}, $\min \Gg_{\sigma''}- \min \Gg_{\sigma'}\geqslant 0$. Hence,
\[\Phi_{\sigma''}-\Phi_{\sigma'}\geqslant (1-\gamma)c_1\nnabla. \]
\noindent \emph{Case B:} By Lemma~\ref{lem:F increase} and Lemma~\ref{lem:G increase},
\begin{eqnarray*}
 \min\Ff_{\sigma''}-\min\Ff_{\sigma'}&\geqslant& c_2\, \dd y,\\
\min\Gg_{\sigma''}- \min\Gg_{\sigma'}&\geqslant& c_3\nnabla-c_4\dd y,
\end{eqnarray*}
 for some constants $c_2,c_3,c_4>0$. Hence,
\begin{eqnarray*}
\Phi_{\sigma''}-\Phi_{\sigma'}&=&(1-\gamma)\left(\min \Ff_{\sigma''}-\min \Ff_{\sigma'}\right)+\gamma\left(\min \Gg_{\sigma''}-\min \Gg_{\sigma'}\right)\nonumber\\
&\geqslant&(1-\gamma) c_2\, \dd y+\gamma (c_3\nnabla-c_4 \dd y)\nonumber\\
&=&\gamma c_3\nnabla + ((1-\gamma)c_2 -\gamma\,c_4) \dd y.\nonumber
\end{eqnarray*}
By choosing $\gamma$ small enough, the constant before $\dd y$ will be positive. We choose $(1-\gamma)=\gamma c_4/c_2$, i.e.
$\gamma=c_2/(c_2+c_4)$. Hence,
\begin{equation}\label{eq:case b} \Phi_{\sigma''}-\Phi_{\sigma'}\geqslant \gamma c_3\nnabla.
\end{equation}
Combining Case A and Case B we obtain
\[\Phi_{\sigma''}-\Phi_{\sigma'}\geqslant \min\left\{(1-\gamma)c_1,\gamma c_3\right\}\nnabla=c_5\nnabla,\]
where \[c_5=\min\left\{(1-\gamma)c_1,\gamma c_3\right\}=\min\left\{\frac{c_1c_4}{c_2+c_4},\frac{c_2c_3}{c_2+c_4}\right\}.
\]
 \end{proof}\bigskip

\noindent \textbf{Proof of Theorem~\ref{th:theorem1}}:\ \ Let $\roo$ be any request sequence. Using Lemma~\ref{lem:nabla increase} and taking
the sum over all requests, we get
\[\Phi_{\roo}- \Phi_{\epsilon}\geqslant c_5\nnabla_{\roo}.\]
Lemma~\ref{lem:phi epsilon} states that $\Phi_{\epsilon}=0$, and Lemma~\ref{lem:phi sigma} states that $\Phi_{\roo}\leqslant \Opt_{\roo}$.
Hence,
\[\nnabla_{\roo}\leqslant \frac{1}{c_5}\ \Opt_{\roo}.\]
By Lemma~\ref{lem:nabla sigma}, the competitive ratio is at most $(1/c_5-1)/\lambda$.

\bigskip

\section{General metric spaces}\label{sec:2-server} In this section, we extend Theorem~\ref{th:theorem1} to arbitrary symmetric metric spaces.
\smallskip

\begin{theorem}\label{th:theorem2}
The work function algorithm $\WFA_{\lambda}$ is constant competitive for the generalized 2-server problem for any constant $\lambda$ with
$0<\lambda<1$.
\end{theorem}

On one hand, the generalization of the proof is easy since all lemmas stay exactly the same, apart from some constants. Moreover, the only proof
that really changes is that of Lemma~\ref{lem:convex containment}. However, to prove this lemma we make the potential function even more complex
than it already is.

A small problem that appears in a discrete metric space is that the new potential function may no longer be a Lipschitz continuous function of the
given request as we stated in Lemma~\ref{lem:G Lipschitz}. To overcome this, we extend the metric space into a metric space
$\overline{\MM}\supseteq \MM$ where any two points are joint by a continuous path, i.e., for any pair $u_1,u_2\in \overline{\MM}$ and $\zeta\in
[0,1]$ there is a point $u_3\in \overline{\MM}$ such that $d(u_1,u_3)=\zeta d(u_1,u_2)$ and $d(u_2,u_3)=(1-\zeta) d(u_1,u_2)$. This can easily be
done and is a common assumption for online routing problems. See for example~\cite{ChLa92} for a discussion on this. We avoid using the notation
$\overline{\MM}$ and simply assume that $\MM$ has this property. Note that this is done only for the analysis. The request sequence and the work
function algorithm will only use points of the original metric space.\smallskip

\subsection{Adjusting the potential}\label{sec:new potential}

\begin{figure}[t]\center
\begin{tikzpicture}
\draw  (-6,0) rectangle (-2,4);
\draw  (0,0) rectangle (4,4);
\draw [red, fill=red!30](-4,2) circle (1.5cm);
\filldraw [red, fill=red!30] (0.5,2) -- (2,0.5)--(3.5,2) -- (2,3.5) -- (0.5,2);
\node at (-3.5,2) {$x_1$}; \node at (-3.5,2.5) {$x_2$};
\fill (canvas cs:x=-4cm,y=2cm) circle (3pt); \fill (canvas cs:x=-4cm,y=2.5cm) circle (3pt);
\node at (2.25,2.25) {$y_1$}; \node at (2.75,2.25) {$y_2$};
\fill (canvas cs:x=2cm,y=2cm) circle (3pt); \fill (canvas cs:x=2.5cm,y=2cm) circle (3pt);
\node at (-6,-0.5) {$(-4,-4)$}; \node at (0,-0.5) {$(-4,-4)$}; \node at (-2,4.5) {$(4,4)$}; \node at (4,4.5) {$(4,4)$}; \node at (-1.5,0) {\large
$\XX$}; \node at (4.5,0) {\large $\YY$};
\end{tikzpicture}
\caption{Example of $\Spheres$. Here, $\XX$ is the Euclidean plane and $\YY$ is the plane with the $L_1$ norm. Further,
$s_1=(x_1,y_1)=((0,0),(0,0))$, and $s_2=(x_2,y_2)=((0,1),(1,0))$. Hence, $d^{\XX}(x_1,x_2)=1$ and $d^{\YY}(y_1,y_2)=1$. If constant $\eta=3$ then
 $\Spheres(s_1,s_2)$ is the Cartesian product of the shaded areas.} \label{fig:spheres}
\end{figure}
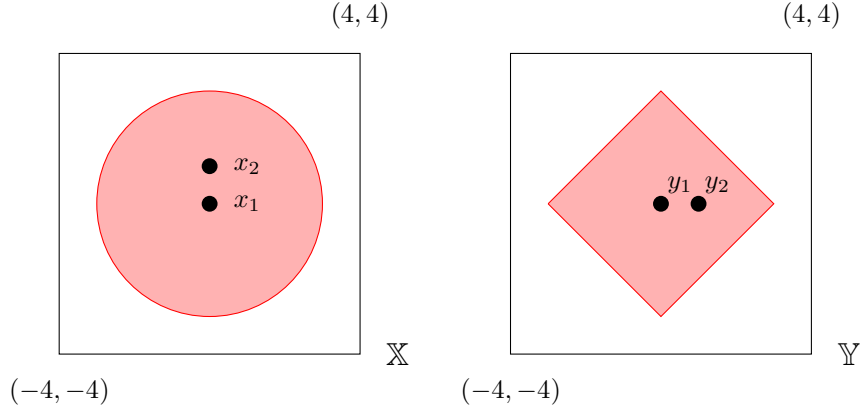
The first point in the \CNN proof where we used the restriction to $\RR^2$ is in the potential function: The set $\Boks(s_1,s_2)$ is defined only
for $s_1,s_2\in \RR^2$. It was defined especially for Lemma~\ref{lem:convex containment} which says that if the points $(s_1,s_2,s_3)$ have the
same $x$- or $y$-coordinate, then one of them is redundant. We applied this in Lemma~\ref{lem:G increase} (Case 1) where we replaced the
redundant point by point $\snabla$. Lemma~\ref{lem:convex containment} still holds for a general metric space but its proof does not apply anymore
because Equality~\eqref{eq:slack equality} is in general an inequality: For any three points $(s_1,s_2,s_3)$ and sequence $\sigma$,
\[\Sl_{\sigma}(s_3;s_1)\le\Sl_{\sigma}(s_3;s_2)+\Sl_{\sigma}(s_2;s_1).\]
Unfortunately, we need $\ge$ here for the proof of Lemma~\ref{lem:convex containment} to hold. Looking ahead at Equation~\eqref{eq:new proof
Sl>Sl+SL} one sees an alternative inequality which takes the place of Equation~\eqref{eq:slack equality}. The trick is quite simple. We make two
changes to the potential function: We add the constraint that $s_3$ should be relatively far from $s_1$ and $s_2$ and we take two different
measures for slack. (See Figure~\ref{fig:two kinds of slack}.) The intuition is that if a point $b$ has a nonnegative slack with respect to a
point $a$, then by using a steeper slack function which has parameter $\mu>\lambda$, the slack of $b$ with respect to $a$ is at least
$(\mu-\lambda)d(a,b)$. We make this precise below.

The following definition takes the place of $\Boks$. Let $\eta\gg 1$.
\begin{eqnarray*}
\Spheres(s_1,s_2)=\{\ (x,y)\in \XX\times\YY  \mid  d^{\XX}(x,x_1)\leqslant \eta \cdot d^{\XX}(x_1,x_2)\ \ \ &&\\
\text{ and } d^{\YY}(y,y_1)\leqslant \eta \cdot d^{\YY}(y_1,y_2)\ \}.&&
\end{eqnarray*}
Note that \Spheres\ is in fact the Cartesian product of a sphere around $x_1$ and a sphere around $y_1$. Instead of $(x_1,y_1)$, we could also
take
 $(x_2,y_2)$ or somehow a point in between. This makes no real difference if $\eta$ is large. (One could think of $\Boks(s_1,s_2)$ as the Cartesian product of a 1-dimensional sphere of diameter $|x_2-x_1|$ around point $(x_1+x_2)/2$ and a 1-dimensional sphere of diameter $|y_2-y_1|$ around point $(y_1+y_2)/2$.)

The other change that we make in the potential function is adjusting the constants. The whole proof for Theorem~\ref{th:theorem1} is still valid
(up to a constant) if we replace the $\lambda$'s that appear in the potential function by some other constant $\mu$ for which $\lambda \leqslant
\mu<1$, while keeping $\WFA_{\lambda}$ the same. (There is no need to verify this claim since we do not use it explicitly.) This freedom in the
parameter leaves way for finetuning the potential function as we will do here. We fix such a $\mu$ with $\lambda < \mu<1$ and define for $s\in
\MM$ and $C\subseteq\MM$ the slack like we did before but now with $\mu$ instead of $\lambda$. In addition, we keep the old definition and add the
parameter $\lambda$ in the notation:
\begin{eqnarray*}\label{eq:slack mu}
\Sl^{\mu}_{\sigma}(s;C)=\min\limits_{t \in C}\{W_{\sigma}(t)+\mu d(s,t)\}-W_{\sigma}(s)\\
\Sl^{\lambda}_{\sigma}(s;C)=\min\limits_{t \in C}\{W_{\sigma}(t)+\lambda d(s,t)\}-W_{\sigma}(s).
\end{eqnarray*}
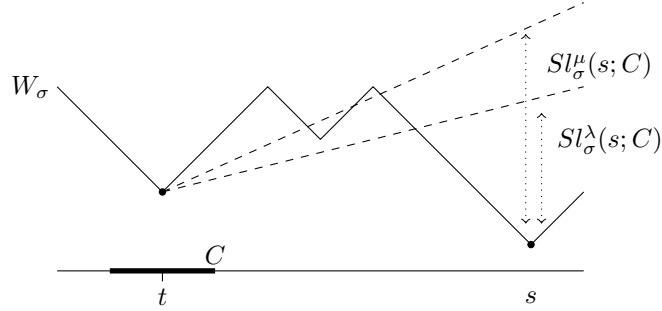
\begin{figure}\label{fig:two kinds of slack}
\center
\begin{tikzpicture}[scale=0.7]
\draw [line width=2pt] (1,0.5)--(3,0.5);
\draw (0,0.5)--(10,0.5);
\draw (2,0.3)--(2,0.5);
\draw (0,4)--(2,2)--(4,4)--(5,3)--(6,4)--(9,1)--(10,2);
\draw [dashed](2,2)--(10,4);
\draw [dashed](2,2)--(10,5.6);
\draw [<->] [dotted](9.2,1.4)--(9.2,3.5);
\draw [<->] [dotted](8.9,1.4)--(8.9,5);
\fill (canvas cs:x=2cm,y=2cm) circle (2pt);
\fill (canvas cs:x=9cm,y=1cm) circle (2pt);
\node at (2,0) {$t$};
\node at (9,0) {$s$};
\node at (3,0.8) {$C$};
\node at (-0.5,4) {$W_{\sigma}$};
\node at (10.5,3) {$\Sl^{\lambda}_{\sigma}(s;C)$};
\node at (10.3,4.4) {$\Sl^{\mu}_{\sigma}(s;C)$};
\end{tikzpicture}
\caption{Two kinds of slack: one with parameter $\lambda$ and one with parameter $\mu>\lambda$.}
\label{fig:double slack}
\end{figure}
Next, we define the new $\Hh_{\sigma}, \Ff_{\sigma}$ and $\Gg_{\sigma}$. For simplicity, we keep the same names although they are now slightly
different functions:
\[\begin{array}{rcl}
\Hh_{\sigma}(s_1,s_2)&=&W_{\sigma}(s_1)-\frac{1}{2}\Sl^{\mu}_{\sigma}(s_2;s_1)\\[3mm]
\Ff_{\sigma}(s_1,s_2,s_3)&=&\Hh_{\sigma}(s_1,s_2)-\beta \Sl^{\lambda}_{\sigma}(s_3;\{s_1,s_2\})\\[3mm]
\Gg_{\sigma}(s_1,s_2,s_3)&=&\Hh_{\sigma}(s_1,s_2)-\beta \Sl^{\lambda}_{\sigma}(s_3;\Spheres(s_1,s_2)).\\
\end{array}
\]
Note that $\mu$ is used for the slack of $s_2$ while $\lambda$ is used for the slack of $s_3$. The potential function is:
\[\Phi_{\sigma}=(1-\kappa)\min\Ff_{\sigma}+\kappa\min\Gg_{\sigma},\]
where $0<\kappa<1$. To prove constant competitiveness, there is no need to specify precise values of the constants. We only need to choose the
constants either large or small enough. The order in which we choose them and the domains are listed below.  For example, given  $\lambda$ and the
choice of $\mu$, there is a number $\eta_0$ such that any choice $\eta\geqslant \eta_0$ is fine for our proof. We do not compute the values
$\eta_0,\beta_0$ or $\kappa_0$ but it will be clear from the proof that such values exist.
\[\begin{array}{rll}
\lambda:&\text{ given parameter},&\\
\mu: & \lambda<\mu<1,&\\
\eta:& \eta\ge\eta_0\gg 1 &\text{, where $\eta_0$ depends on $\lambda$ and $\mu$,} \\
\beta:&0<\beta\leqslant \beta_0< 1/2 &\text{, where $\beta_0$ depends on $\lambda,\mu$ and $\eta$,} \\
\kappa:&0<\kappa<\kappa_0<1&\text{, where $\kappa_0$ depends on $\lambda,\mu,\eta$ and $\beta$.}
\end{array}
\]

\subsection{Adjusting the proofs of the lemmas}
All lemmas stay exactly the same apart from some constants. Moreover, all proofs stay basically the same. The only proof that is really different
is that of Lemma~\ref{lem:convex containment}. Let us go over all the lemmas one by one.

Nothing changes for Section~\ref{sec:prelimanaries} since it comes before the potential function and holds for any metric space. The first lemma
in Section~\ref{sec: cnn problem} is Lemma~\ref{lem:s dominates t}. The lemma holds with different constants. The bounds we get are
\begin{equation}\label{eq:new bounds}
\begin{array}{lll}
(a),(d)&:& \delta\cdot\beta(1-\lambda)\\
(b),(c)&:& \delta\cdot\left(\frac{1}{2}(1-\mu)-\beta(1+\lambda)\right)\\
(e)&:& \delta\cdot\left(\frac{1}{2}(1-\mu)-\eta\beta(1+\lambda)\right)\\
(f)&:& \delta\cdot\left(\frac{1}{2}(1-\mu)-(\eta+1)\beta(1+\lambda)\right)\\
\end{array}
\end{equation}

The proof for (a),(b),(c),(d) remains the same, only $\alpha$ becomes $\beta$ and some of the $\lambda$'s become $\mu$.  In (e), there is an
additional factor $\eta$ because a move of $s_2$ over some distance may cause the border of $\Spheres(s_1,s_2)$ to move by $\eta$ times this
distance. For a move of $s_1$, this factor is $\eta+1$ since \Spheres\ is defined around $s_1$. The precise bounds are not so important. We only
need to see that we can choose $\beta$ small enough to let all the right hand sides be $\Omega(\delta)$.

Nothing changes for Lemma~\ref{lem:minimum in r}. In the proof of Lemma~\ref{lem:F<=G<=H}, only $\Boks$ needs to be replaced by $\Spheres$. In the
proof of Lemma~\ref{lem:cardinality=2}, we only need to update the definition of $\Ff$. \bigskip

\noindent\textbf{New proof of Lemma~\ref{lem:convex containment}}. \begin{proof} Let $s_1,s_2,s_3$ have the same $y$-coordinate. We may assume
that
\begin{equation}\label{eq:new proof >0}
\Sl^{\lambda}_{\sigma}(s_3;\Spheres(s_1,s_2))> 0,
\end{equation}
since otherwise $\Hh_{\sigma}(s_1,s_2)\leqslant \Gg_{\sigma}(s_1,s_2,s_3)$ and we are done. By this assumption, we have $s_3\notin
\Spheres(s_1,s_2)$.  Hence, $d(s_1,s_3)> \eta d(s_1,s_2)$ (using $d^{\XX}(s_i,s_j)=d(s_i,s_j)$ for $i,j\in\{1,2,3\}$). Then
\begin{eqnarray}\label{eq:new proof eq1}
\Sl^{\mu}_{\sigma}(s_3;s_1)&=&W_{\sigma}(s_1)+\mu d(s_1,s_3)-W_{\sigma}(s_3)\nonumber \\
&=&W_{\sigma}(s_1)+\lambda d(s_1,s_3)-W_{\sigma}(s_3)+(\mu-\lambda)d(s_1,s_3)\nonumber \\
&=&\Sl^{\lambda}_{\sigma}(s_3;s_1)+(\mu-\lambda)d(s_1,s_3)\nonumber \\
&>&\Sl^{\lambda}_{\sigma}(s_3;s_1)+\eta(\mu-\lambda)d(s_1,s_2).
\end{eqnarray}
By choosing $\eta$ large enough (given the values of $\lambda$ and $\mu$), we guarantee that $\eta(\mu-\lambda)\geqslant 1+\mu$. If we also use
that $(1+\mu)d(s_1,s_2)\geqslant \Sl^{\mu}_{\sigma}(s_2;s_1)$ (follows directly from~\eqref{eq:slack mu}) then the analogue of~\eqref{eq:slack equality} becomes
\begin{equation}\label{eq:new proof Sl>Sl+SL}
\Sl^{\mu}_{\sigma}(s_3;s_1)>\Sl^{\lambda}_{\sigma}(s_3;s_1)+\Sl^{\mu}_{\sigma}(s_2;s_1).
\end{equation}
The remainder of the proof is similar to the original proof.  For the first two inequalities below we use, respectively, \eqref{eq:new proof
Sl>Sl+SL} and Lemma~\ref{lem:subsetslack}. For the last inequality we use~\eqref{eq:new proof >0} and $\beta<  1/2$.
\begin{eqnarray*}
\Hh_{\sigma}(s_1,s_3)&=&W_{\sigma}(s_1)-\frac{1}{2}\Sl^{\mu}_{\sigma}(s_3;s_1)\\
&< &W_{\sigma}(s_1)-\frac{1}{2}\Sl^{\mu}_{\sigma}(s_2;s_1)-\frac{1}{2}\Sl^{\lambda}_{\sigma}(s_3;s_1)\\
&\leqslant &W_{\sigma}(s_1)-\frac{1}{2}\Sl^{\mu}_{\sigma}(s_2;s_1)-\frac{1}{2}\Sl^{\lambda}_{\sigma}(s_3;\Spheres(s_1,s_2))\\
&<&W_{\sigma}(s_1)-\frac{1}{2}\Sl^{\mu}_{\sigma}(s_2;s_1)-\beta\Sl^{\lambda}_{\sigma}(s_3;\Spheres(s_1,s_2))\\
&=  &\Gg_{\sigma}(s_1,s_2,s_3).\end{eqnarray*}

\end{proof}
The proof of Lemma~\ref{lem:phi epsilon} stays the same. Also the proof of Lemma~\ref{lem:phi sigma} stays the same apart from $\alpha$ becoming
$\beta$ and $\gamma$ becoming $\kappa$. Lemmas~\ref{lem:s in r dominted by} and~\ref{lem:nabla dx} do not depend on the potential function, nor
do they depend on the metric space. These lemmas and proofs stay exactly the same. Lemma~\ref{lem:G Lipschitz} was given without proof and again
it can  easily be verified from the definitions. The proof of Lemma~\ref{lem:domination} does not change. For Lemma~\ref{lem:F increase} there are
a few small changes. In Case A, only $\alpha$ becomes $\beta$. In Case B, the last inequality is different since the inequalities of
Lemma~\ref{lem:s dominates t} are different. The new values were given in formula~\eqref{eq:new bounds}. All we need to notice is that by choosing
$\beta$ small enough (depending on $\lambda,\mu$ and $\eta$), the righthand sides are $\Omega(\delta)$.

Also in the proof of Lemma~\ref{lem:G increase} there are a few small changes. Of course, $\alpha$ becomes $\beta$ and $\Boks$ becomes $\Spheres$.
The new function $\Gg$ is still Lipschitz continuous. Hence, we may assume $y'=y''$. We consider the same three cases and the proof for the first
and second case remain the same. For Case 3 we need to use the new bounds of Lemma~\ref{lem:s dominates t}. Then, Equation~\eqref{eq:case3 1}
becomes
\[ \Gg_{\sigma'}(s_1,s_2,s_3)\geqslant\min\Gg_{\sigma'}+\left(\frac{1}{2}(1-\mu)-(\eta+1)\beta(1+\lambda)\right)\dd x.\]
Combining this with Equation~\eqref{eq:case3 2} as we did, we get
\[
\min\Gg_{\sigma''}-\min\Gg_{\sigma'}\geqslant\left(\frac{1}{2}(1-\mu)-(\eta+1)\beta(1+\lambda)-2\beta\right)\dd x.
\]
By choosing $\beta$ small enough (depending on $\lambda,\mu$ and $\eta$), the righthand side is at least $c_3\nnabla$ for some constant $c_3$
(using Lemma~\ref{lem:nabla dx}). The proof of Lemma~\ref{lem:minG''ge minG'} remains the same, Finally, the only change in the proof of
Lemma~\ref{lem:nabla increase} is that $\gamma$ becomes $\kappa$.
\bigskip

\section{A decomposition approach for the generalized $k$-server problem}\label{sec:higher dim}

The generalized $k$-server problem appears a lot more complicated for dimensions $k\geqslant 3$. It is unclear if for any fixed $k\geqslant 3$ a constant
competitive ratio $f(k)$ is possible at all. In any case, the ratio will be at least $k^{\Omega(k)}$~\cite{FiaRic94}. The question is important
for its relation to sum problems discussed in the introduction. Interestingly, the proof for $k=2$ does show a decomposition into subproblems
which can be generalized to any $k$ and which seems to be a real simplification of the problem. Although an answer to these subproblems is
missing, it does give an example of decomposing a sum problem into (apparently) easier problems.

Suppose that we can find $k$ functions $\Ff^{(i)}_{\sigma}: \MM^{k+1}\rightarrow \RR$, for $i=1,2,\dots,k$ with the following two properties:
\begin{enumerate}
\item[(i)] $\min\Ff^{(i)}_{\epsilon}=0$ and $\min\Ff^{(i)}_{\sigma}\leqslant \Opt_{\sigma}$ (where $\epsilon$ is the empty sequence)

\item[(ii)] Let $r'=r(x'_1,x'_2,\dots,x'_k)$ and $r''=r(x''_1,x''_2,\dots,x''_k)$ be two subsequent requests and let $\pi_1,\pi_2,\dots,\pi_k$ be a
permutation of $1,2,\dots,k$ such that  $\dd x_{\pi_1}\leqslant \dd x_{\pi_2}\leqslant \dots \leqslant \dd x_{\pi_k}$, where $\dd_{\pi_i}=|
x''_{\pi_i}-x'_{\pi_i}|$. Further, denote $\sigma'=\sigma r'$ and $\sigma''=\sigma r' r''$ and denote the extended cost
$\nnabla_{r''}(W_{\sigma'})$ simply by $\nnabla$. Then for all $i$ there are constants $a^{(i)},b^{(i)},c^{(i)},d^{(i)}>0$ (depending on $k$
and $\lambda$) such that either (A) or (B) holds:
\begin{enumerate}
\item[(A)] $\min \Ff^{(i)}_{\sigma''}-\min \Ff^{(i)}_{\sigma'}\geqslant a^{(i)}\nnabla-b^{(i)}\dd x_{\pi_{i-1}}$ and for all $j>i$ \\ $\min
\Ff^{(j)}_{\sigma''}-\min \Ff^{(j)}_{\sigma'}\geqslant 0$,
\item[(B)] $\min \Ff^{(i)}_{\sigma''}-\min \Ff^{(i)}_{\sigma'}\geqslant c^{(i)}\dd
x_{\pi_i}$.
\end{enumerate}
\end{enumerate}
Then, the following potential function proves that $\WFA_{\lambda}$ is constant competitive for some constant depending on $k$ and $\lambda$.
\[\Phi_{\sigma}=\sum_{i=1}^{k}\gamma^{(i)}\cdot \min_{s_1,\dots,s_{k+1}\in\MM}\Ff^{(i)}_{\sigma}(s_1,s_2,\dots,s_{k+1}),\] where
$\gamma^{(1)}+\dots+\gamma^{(k)}=1$ and $\gamma^{(i)}/\gamma^{(i+1)}\geqslant b^{(i+1)}/c^{(i)}$ for $i=1,2,\dots,k-1$.

First, let us see how this relates to our proof for $k=2$. We denoted $\Ff^{(1)}=\Ff$ and $\Ff^{(2)}=\Gg$  and denoted $x_1$ and $x_2$ by $x$ and
$y$. We assumed $\dd y\leqslant \dd x$ which implies $\pi_1=2$ and $\pi_2=1$. Property (i) holds (and was used for Lemma~\ref{lem:phi epsilon} and
Lemma~\ref{lem:phi sigma}). Now, it is easy to check that property (ii) corresponds with Lemma~\ref{lem:F increase}, \ref{lem:G increase}, and
\ref{lem:minG''ge minG'}. (Define $\dd x_{\pi_{0}}:=0$ and note that $\nnabla=O(\dd x_k)$.)

Next, we give a short sketch why this would give a proof of competitiveness and then we argue why this is an interesting decomposition. We  need to
show that the increase in the potential for the new request $r''$ is at least some constant times $\nnabla$. (Then, if additionally
$\Phi_{\epsilon}=0$ and $\Phi_{\sigma}\leqslant \Opt_{\sigma}$, competitiveness follows from Lemma~\ref{lem:nabla sigma}.)  First consider $i=1$.
(Define $\dd x_{\pi_{0}}:=0$.) If case (A) applies then we are done.  So assume from now that (B) applies for $i=1$. Consider $i=2$. If case (A)
applies for $i=2$ then, by the choice of the $\gamma^{(i)}$, the increase in the potential function is at least
\[\gamma^{(1)}c^{(1)}\dd x_{\pi_{1}}+\gamma^{(2)}(a^{(2)}\nnabla- b^{(2)}\dd x_{\pi_{1}})\geqslant \gamma^{(2)}a^{(2)}\nnabla. \]
So assume case (B) applies and consider $i=3$. We can repeat the argument  until finally we consider $i=k$. Then the proof follows from case B as
well since $\dd x_{\pi_k}=\Omega(\nnabla)$.

Now we will argue that the decomposition seemingly simplifies the analysis. Remember that the general idea is to find a potential function
$\Phi_{\sigma}$ with the property that the increase for every new request is at least some constant (depending on $k$) times the extended cost
$\nnabla$ of the new request. In the decomposition, this property is split into $k$ weaker properties. Assume that the functions $\Ff^{(i)}$ are
all Lipschitz continuous functions of the last request. By this we mean, if $r''$ is changed to some other request $\tilde{r}''$ while keeping
the arguments $s_1,\dots,s_{k+1}$ fixed, then the value $\Ff^{(i)}_{\sigma''}(s_1,\dots,s_{k+1})$ changes by at most some constant (depending on
$k$ and $\lambda$) times $||r''-\tilde{r}''||$. Lipschitz continuity seems a natural property. Note that the extended cost
$\nnabla:=\nnabla_{r''}(W_{\sigma'})$ is always Lipschitz continuous in $r''$. If Lipschitz continuity holds, then to prove (ii), we may assume that
$\dd x_{\pi_h}=0$ for all $h<i$, as we did in the proof of Lemma~\ref{lem:G increase}. For example, for $\Ff^{(k)}_{\sigma}$ we only need to show
an increase of $\Omega(\nnabla)$ under the (strong) condition that $\dd x_j=0$ for all $j\leqslant k-1$, i.e., under the condition that only one
coordinate changes.

We will not speculate on a general decomposition theorem for sum problems and merely say that the outline appears a significant step towards a proof
for $k\geqslant 3$ and is an interesting  contribution towards a general theory of competitiveness of metrical service systems.

\bigskip

\section{Notes and open problems}\label{sec:future}
The most prominent research direction is to enhance the theory of competitiveness of metrical service (or task) systems and in particular for the
generalized work function algorithm. Our proof shows that only very limited information of the work function may be needed to show that
$\WFA_{\lambda}$ performs well.  In fact, we only used the obvious properties that apply to any work function, e.g. that any point $s\in \MM$ is
dominated by some point $t$ on the last request and that the work function is Lipschitz continuous with constant 1. (As a comparison, Koutsoupias
and Papadimitriou show for their $k$-server proof that the $k$-server work function has some nice quasiconvexity property.) So, if this is all we
use, why does not this imply competitiveness for any metrical service system? The answer is that the potential function was designed for the
typical requests of the generalized 2-server problem, i.e., the potential function exploits that the support of any work function is a subset of
the last request. This kind of analysis, that is purely based on the geometry of a single request, is interesting for metrical service systems in
general. For this purpose, our potential function has some valuable ingredients such as the use of convex sets like \Boks\ and \Spheres\ and the
use of slack functions with different parameters ($\lambda$ and $\mu$). These techniques are helpful for isolating extreme solutions, i.e., (a
small number of) solutions which in a way represent all offline solutions.

An illustrative example is the problem of chasing lines. In this system, the metric space is $\RR^d$ and the set $\mathcal{R}$ of requests
contains all lines and line segments in $\RR^d$. By taking our function $\Gg$ as the potential function (where $\Boks(s_1,s_2)$ is now defined as
the line segment between $s_1$ and $s_2$), it follows immediately that $\WFA_{\lambda}$ is constant competitive (independent of $d$) for any
$\lambda\in (0,1)$. All that we need to notice is the following alternative formulation of Lemma~\ref{lem:convex containment}: If $s_1,s_2,s_3\in
\RR^d$ are all on a straight line, then $\Hh_{\sigma}(u_1,u_2)\leqslant \Gg_{\sigma}(s_1,s_2,s_3)$ for some $u_1,u_2\in \{s_1,s_2,s_3\}$. Now
assume that sequence $\sigma$ is followed by a request $r$ and that $s_1,s_2,s_3$ minimize $\Gg_{\sigma r}$. Then, all three points are on the
last request $r$ and hence all are on a straight line. The lemma says that one of the three points is redundant. Replacing one of the three points
 by a point $\snabla\in \MM$ with maximum extended cost $\nnabla$, we see that the increase for the potential function is $\Omega(\nnabla)$ and
competitiveness follows. The algorithm by Friedman and Linial~\cite{FriedmanL:1993} for line chasing is much less general and uses angles and
coordinates in the Euclidean plane. Of course, how one can implement $\WFA_{\lambda}$ efficiently for the line chasing problem is a different
story.

\subsection{Open problems}
There are some very intriguing open problems in online optimization. Examples are the \emph{$k$-server conjecture} (deterministic and randomized)
and the \emph{dynamic optimality conjecture}~\cite{SleatorTarjan1985} for binary search trees. (We refer to~\cite{DemaineHIKP2009} for a survey of
recent results.) The latter conjecture states that there exists a constant competitive online algorithm for binary search trees. Maybe not so
well-known is that the binary search tree problem (without insertions or deletions) can be transformed into a metrical service system with loss of
a constant factor in the approximation. This can be done as follows. Let $1,2,\dots,n$ be the items of the tree. By a binary search tree, we mean a
rooted tree with maximum degree three. Then, the metric space consists of all binary search trees with nodes $1,2,\dots,n$ and the distance
between two trees is the minimum number of rotations needed to transform one tree into the other (or we may take any other distance functions that
is within a constant factor). Now, for each item $i$ we define a request $r^{i}$ which is the set of all binary search trees with root $i$. The
collection of possible requests is $\mathcal{R}=\{r^{1},r^{2},\dots,r^{n}\}$. Let $\BST$ be the  binary search tree problem (as defined
in~\cite{SleatorTarjan1985}) and let $\BST^*$ be the $\BST$ problem modeled as a metrical service system as described above . The next theorem
states that the (online) approximation ratios of these two problems are within a constant factor. A similar result is given
in~\cite{BorodinElYaniv1998Book} (Lemma 1.3) for the list update problem, which is the one-dimensional equivalent of the \BST\ problem. There, it
is shown that any $c$-competitive algorithm remains $c$-competitive if the cost to serve $i$ is depth$(i)-1$ in stead of depth$(i)$. Below, we
also show the other direction and model it as a metrical service system.

\smallskip\begin{lemma}\label{lem:BST and BST*} The (online) approximation ratios of the $\BST$ problem and its metrical service system formulation $\BST^*$ are
within a constant factor.
\end{lemma}
\begin{proof}
In~\cite{SleatorTarjan1985}, the cost for serving an item $i$ is one plus depth$(i)$, the depth of item $i$ in the current tree. This differs with
our service system model in two ways. First, in our model there is the restriction that $i$ has to be moved to the root in order to serve it. Note
that this restriction only increases the cost by a small constant factor since $i$ can be moved to the root and back at a cost
$O(\text{depth}(i))$. The other difference is that in our model there is no additive cost of one to serve a request. In particular, that means
that items at the root are served at a cost of one in the $\BST$ model while these are for free in $\BST^*$. We call $\BST^*$ the \emph{zero cost
model}. Next, we compare the competitive ratios for $\BST$ and $\BST^*$ with the restriction that items can only be served at the root. Under this
restriction, let $\Opt$ and $\Opt_0$ denote the optimum in respectively the standard cost and the zero cost model. Then for any sequence $\sigma$:
$\Opt(\sigma)=\Opt_0(\sigma)+|\sigma|$. Let $\Alg$ be any $c$-competitive algorithm for $\BST^*$. Then, it is $c$-competitive for $\BST$ as well
(See also Lemma 1.3 in~\cite{BorodinElYaniv1998Book}):
\[\Alg(\sigma)=\Alg_0(\sigma)+|\sigma| \leqslant c\Opt_0(\sigma)+|\sigma|= c\Opt(\sigma)-(c-1)|\sigma|.\]

For the other direction, assume that some algorithm $\Alg$ is $c$-competitive for $\BST$. We will show that this gives a $(2c-1)$-competitive
algorithm for $\BST^*$. For any request sequence $\sigma$ we define $\sigma'$ as the sequence obtained by removing the repeated requests. For
example, if item $i$ is requested three times consecutively then we remove two of these. Now define algorithm $\Alg'$ as follows. For any request
sequence $\sigma$ it gives the truncated sequence $\sigma'$ to $\Alg$ and then behaves exactly like $\Alg$. This means that when a requested item
$i$ is moved to the root, then the search tree remains unchanged until the first moment that a different item is requested. This way, sequence
$\sigma$ is served using the online solution for $\sigma'$. By assumption, $\Alg(\sigma')\leqslant c\Opt(\sigma')$. Further, if we assume that the first
request is not to the root then $|\sigma'|\leqslant \Opt_0(\sigma')$.
\begin{eqnarray*}
&\Alg'_0(\sigma)=\Alg'_0(\sigma')= \Alg(\sigma')-|\sigma'|\leqslant c \Opt(\sigma')-|\sigma'|= \\
&c\Opt_0(\sigma')+(c-1)|\sigma'|\leqslant (2c-1)\Opt_0(\sigma')=(2c-1)\Opt_0(\sigma).
\end{eqnarray*}
\end{proof}\smallskip

The \BST\ problem is still not well-understood. It is not known if the problem is NP-hard, nor is there a constant factor offline approximation
algorithm known. Lower bounds on the optimal solution are hard to get. However, if constant competitiveness is possible, then probably there is no
need  for this kind of bounds. In online optimization the analysis is usually based on some kind of extreme solutions that in a way represent all
possible offline solutions. A simple (and highly relevant) example is the list update problem~\cite{BorodinElYaniv1998Book}. The move-to-front
rule has optimal competitive ratio of $2-2/(n\! +\! 1)$, where $n$ is the size of the list. It is easy to see that it is $2$-competitive since
with loss of a factor 2 we may assume that each item can only be served at the front. But then, there is only one optimal solution and the
move-to-front algorithm gets one step closer to the optimal solution with every rotation that it makes. The only information about the optimal
offline solution that is used in this analysis is that its current configuration serves the current request. Hence, for list-update it is enough
to consider only one offline solution. Another example is the (optimal) double coverage algorithm for the $k$-server problem on
trees~\cite{ChrLar91A} where the potential function is defined only by the current configuration of the online solution and that of the optimal
solution. An example with two extreme solutions is the line chasing problem that we discussed at beginning of this section. We sketched a proof
with a potential function which is defined by three solutions. We could show that one of these was redundant and the proof followed easily. Hence,
for $\WFA_{\lambda}$ applied to line chasing there are only two extreme solutions. This analysis follows purely from the geometry of a single
request. There is no need for lower bounds on a sequence of requests. More complicated examples are the $2k-1$ ratio for the $k$-server
problem~\cite{KouPap95A} with a potential based on $k+1$ configurations, and our potential which uses six configurations.

It is not hard to show that $\WFA_{\lambda}$ is in fact not constant competitive for binary search trees when we define the metric space as in
Lemma~\ref{lem:BST and BST*}. However, all kinds of variations are possible. Consider the following adjustment of the metric space. The cost of a
single rotation remains one but the cost of a splaying operation on item $i$ is only some small constant times depth$(i)$. This way,
$\WFA_{\lambda}$ will behave much like the splay tree algorithm and it seems a good candidate for being constant competitive.

A question that pops up is whether such an approach with an adjusted metric has potential at all since we just noted that $\WFA_{\lambda}$ is not
competitive for the natural distance function. Is $\WFA_{\lambda}$ robust in the sense that small changes in the metric give small changes in the
competitive ratio of $\WFA_{\lambda}$? In that case our suggested approach is doomed to fail. Fortunately, the answer is negative and follows from
the next example.

\begin{example}
Consider a metrical service system on a star graph with $k$ leaves. Let $c$ be the center and let $U=\{u_1,u_2,\dots,u_k\}$ be the set of leaves.
The distances are  $d(c,u_1)=1-\epsilon$ and $d(c,u_i)=1$, $i=2,\dots,k$. The set of requests is $\mathcal{R}=\{\{c\}, \{U\setminus
u_2\},\{U\setminus u_3\},\dots,\{U\setminus u_k\}\}$. The optimal online algorithm moves to $c$ whenever the request is $\{c\}$ and moves to $u_1$
otherwise. The competitive ratio of this algorithm is 1. The work function algorithm $\WFA_{\lambda}$ behaves the same for any $\lambda\in(0,1)$
and therefore has ratio 1 as well. If we now change $d(c,u_1)$ from $1-\epsilon$ to $1+\epsilon$, then the optimal online algorithm stays the same
and now has competitive ratio $1+\epsilon$. However, $\WFA_{\lambda}$ can be forced to visit all $u_i$ between two requests for $c$ and has ratio
$(k+\epsilon)/(1+\epsilon)$.
\end{example}\smallskip

An obvious drawback of the work function approach for the \BST\ problem is that it is computationally expensive. In fact, no polynomial time
constant factor approximation is known. Nevertheless, at the moment it is very interesting to see if a constant competitive algorithm is possible
at all, no matter how high the running time.

Below, we list some interesting open problems related to this paper, starting with the $\BST$ problem discussed above.

\begin{itemize}
\item[$\diamond$] Give a constant competitive work function based algorithm for binary search trees (without insertions or deletions).
Although the algorithm would be inefficient it would clearly be a big step towards proving competitiveness for more efficient algorithms
 like splaying.

 \item[$\diamond$] Prove or disprove that the generalized $k$-server problem or weighted $k$-server problem has an $f(k)$-competitive algorithm for some function $f(k)$.
Same for the randomized problem. An outline for a possible proof is given in Section~\ref{sec:higher dim}.

\item[$\diamond$] What is the competitive ratio of the $k$-point request problem, introduced in~\cite{PapadimitriouYannakakis1991}?
Fiat et al.~\cite{FiatFKRRV1998} gave an upper bound of $O(9^k)$ which was improved by Burley~\cite{Bu96} who showed that the generalized work
    function algorithm is $O(k2^k)$-competitive. The best known lower bound is $\Omega(2^k)$~\cite{FiatFKRRV1998}. Just as Burley we
    conjecture that $O(2^k)$ is possible. A good candidate seems a dynamic work function algorithm: one that adjusts the parameter $\lambda$
    online. Such a dynamic work function algorithm would be more powerful than the generalized work function algorithm. Randomization reduces
    the ratio drastically as shown by Ramesh~\cite{Ramesh1993} who gave an upper of $\Omega(k^{13})$ against a lower bound of
    $k/2$~\cite{FiatFKRRV1998}.

\item[$\diamond$] What is the competitive ratio of the continuous \CNN problem?
A lower bound of 3 and upper bound of 6.46 is given in~\cite{AugustineGravin2010}.

\item[$\diamond$] Give other examples of natural metrical service systems that have a constant competitive ratio.
For example, Friedman and Linial~\cite{FriedmanL:1993} give a competitive algorithm if the requests are convex subset of $\RR^2$. They
conjecture that the same applies to $\RR^d$ for any fixed $d$ and  show that it is enough to prove this for affine half spaces. In the
beginning of this section we sketched a proof that $\WFA_{\lambda}$ is constant competitive for lines in $\RR^d$ and it would be interesting
to extend this to half spaces.

\item[$\diamond$] The $k$-server problem has some simple special cases for which $2k-1$ is still the best known ratio, for example the 3-server
problem and the $k$-server problem on a cycle: Find an algorithm with a smaller ratio. The ratio for trees is $k$ but it is unknown if the
work function algorithm achieves this ratio. See~\cite{Koutsoupias2009} for more background on this.

\item[$\diamond$] What is the competitive ratio of the weighted $k$-point request problem, discussed in~\cite{ChrobakSgall2004}? This problem is
a special case of the generalized $k$-server problem and a generalization of the $k$-point request problem.

\item[$\diamond$] Extend the theory of sum problems. For example, by analyzing the sum problem of another elementary metrical task system.

\item[$\diamond$] Prove (or disprove) that the generalized work function algorithm $\WFA_{\lambda}$ is $O(\log n)$-competitive for the online matching problem on a
line. A lower bound of $\Omega(\log n)$ and an upper bound of $O(n)$ were given in~\cite{KoutsoupiasNanavati2003}.
\end{itemize}

\bibliographystyle{siam}

\end{document}